\documentclass[11pt]{article} 
\usepackage{fullpage}
\usepackage{graphicx}
\usepackage{upref}
\usepackage{enumerate}
\usepackage{latexsym}
\usepackage{pdfpages}
\usepackage[bookmarks]{hyperref}
\usepackage{color,graphics}
\usepackage{comment} 
\usepackage{caption}
\usepackage{subcaption}
\usepackage{wrapfig}
\usepackage{braket}

\usepackage{amssymb}
\usepackage{amsmath}
\usepackage{amsthm}
\usepackage{mathrsfs}
\usepackage{mathtools}

\usepackage{epsfig,graphicx}
\usepackage{algorithmic}
\usepackage[ruled,linesnumbered]{algorithm2e}
\usepackage{amsfonts}
\usepackage{tikz} 
\usepackage{varioref}
\usepackage[ansinew]{inputenc}
\usepackage{qcircuit}
\usepackage{authblk}
\usepackage{multirow}
\usepackage{lineno}

\theoremstyle{plain}
\newtheorem{theorem}{Theorem}[section]
\theoremstyle{plain}
\newtheorem{lemma}{Lemma}[section]
\theoremstyle{plain}

\newtheorem{corollary}{Corollary}[section]
\theoremstyle{plain}

\theoremstyle{plain}
\newtheorem{conjecture}{Conjecture}
\theoremstyle{plain}
\newtheorem{claim}{Claim}[section]

\theoremstyle{definition}

\newtheorem{definition}{Definition}[section]
\theoremstyle{definition}
\newtheorem{fact}{Fact}[section]

\theoremstyle{remark}

\theoremstyle{definition}

\AtBeginDocument{}

\DeclareMathOperator{\real}{\mathbb{R}}
\DeclareMathOperator{\nat}{\mathbb{N}}
\newcommand{\cmplx}{\mathbb{C}}
\newcommand{\intg}{\mathbb{Z}}

\newcommand{\poly}{\text{poly}}

\newcommand{\tr}{\text{Tr}}

\newcommand{\id}{\mathbb{I}}

\newcommand{\cliff}{\mathcal{C}}
\newcommand{\pauli}{\mathcal{P}}
\newcommand{\clifft}{\mathcal{J}}

\newcommand{\X}{\text{X}}
\newcommand{\Y}{\text{Y}}
\newcommand{\Z}{\text{Z}}
\newcommand{\had}{\text{H}}
\newcommand{\T}{\text{T}}
\newcommand{\CNOT}{\text{CNOT}}
\newcommand{\phase}{\text{S}}

\newcommand{\tcount}{\mathcal{T}}
\newcommand{\teps}{\mathcal{T}_{\epsilon}}
\newcommand{\tdepth}{\mathcal{T}_{d}}
\newcommand{\tdeps}{\mathcal{T}_{d\epsilon}}

\newcommand{\conj}[1]{\overline{#1}}

\newcommand{\dph}{D_P}
\newcommand{\tset}{\overline{T}}
\newcommand{\genv}{\mathbb{V}}
\newcommand{\diag}{\text{diagonal}}

\begin{document}
\title{T-count and T-depth of \emph{any} multi-qubit unitary}

\author[1,2]{Vlad Gheorghiu \thanks{vlad.gheorghiu@uwaterloo.ca}}
\author[1,2,3,4]{Michele Mosca \thanks{michele.mosca@uwaterloo.ca}}
\author[1,3]{Priyanka Mukhopadhyay \thanks{Corresponding author : mukhopadhyay.priyanka@gmail.com, p3mukhop@uwaterloo.ca}}

\affil[1]{Institute for Quantum Computing, University of Waterloo, Waterloo ON, Canada}
\affil[2]{softwareQ Inc., Kitchener ON, Canada}
\affil[3]{Dept. of Combinatorics and Optimization, University of Waterloo, Waterloo ON, Canada}
\affil[4]{Perimeter Institute for Theoretical Physics, Waterloo ON, Canada}

\date{}
\maketitle

\begin{abstract}
 While implementing a quantum algorithm it is crucial to reduce the quantum resources, in order to obtain the desired computational advantage. For most fault-tolerant quantum error-correcting codes the cost of implementing the non-Clifford gate is the highest among all the gates in a universal fault-tolerant gate set. 
 
 In this paper we design an algorithm to determine the (minimum possible) T-count of any $n$-qubit ($n\geq 1$) unitary $W$ of size $2^n\times 2^n$, over the Clifford+T gate set. The space and time complexity of our algorithm are $O\left(2^{2n}\right)$ and $O\left(2^{2n\teps(W)+4n}\right)$ respectively. $\teps(W)$ (\emph{$\epsilon$-T-count}) is the (minimum possible) T-count of an exactly implementable unitary $U$ i.e. $\tcount(U)$, such that $d(U,W)\leq\epsilon$ and $\tcount(U)\leq\tcount(U')$ where $U'$ is any exactly implementable unitary with $d(U',W)\leq\epsilon$. $d(.,.)$ is the global phase invariant distance. 
 Our algorithm can also be used to determine the (minimum possible) T-depth of any multi-qubit unitary and the complexity has exponential dependence on $n$ and \emph{$\epsilon$-T-depth}. This is the first algorithm that gives T-count or T-depth of \emph{any} multi-qubit ($n\geq 1$) unitary. For small enough $\epsilon$, we can synthesize the T-count and T-depth-optimal circuits. 
 
 Our results can be used to determine the minimum count (or depth) of non-Clifford gates required to implement any multi-qubit unitary with a finite universal gate set consisting of Clifford and non-Clifford gates like Clifford+CS, Clifford+V, etc. To the best of our knowledge, there were no such optimal-synthesis algorithms for arbitrary multi-qubit unitaries in any universal gate set, considering any distance metric.

\end{abstract}


\section{Introduction}

The vision of Feynman \cite{1982_F} that a quantum computer can be used to overcome the limitations of classical computers, gained momentum with the design of quantum algorithms that outperform their classical counterparts for  popular challenging problems like factorization \cite{1994_S,1999_S}, searching an unstructured solution space \cite{1996_G}. Quantum circuit is one of the most popular way for describing and implementing quantum algorithms. These consist of a series of elementary operations or gates belonging to a universal set and dictated by the implementing technologies. Like their classical counterpart, circuit synthesis and optimization is a significant part of any quantum computer compilation process. A compiler primarily translates from a human readable input (programming language) into instructions that can be executed directly on a hardware. An integral part of this process is \textbf{quantum circuit synthesis}, whose aim is to decompose a unitary operation into a sequence of gates from a universal set. Often, additional constraints are imposed on a synthesis task, like minimization of a certain resource like qubits, gates (total number), non-Clifford gates, multi-qubit gates, etc. We call them \textbf{resource-optimal synthesis algorithm}.

Our work primarily focuses on the ``Clifford+T'' gate set, a popular finite universal fault-tolerant set. Fault-tolerant quantum error correction \cite{2000_ZLC, 2005_BK} is required to control the accumulation of errors due to noise on quantum information, faulty quantum gates, faulty quantum state preparation, faulty measurements, etc. This is especially important for long computations, else the errors will make negligible the likelihood of obtaining a reliable and useful answer. The non-Clifford T gate has known constructions in most of the error correction schemes and the cost of fault-tolerantly implementing it exceeds the cost of the Clifford group gates by as much as a factor of hundred or more \cite{2009_FSG, 2006_AGP}. 
The minimum number of T-gates required to implement certain unitaries is a quantifier of difficulty in many algorithms \cite{2016_BG, 2016_BSS} that try to classically simulate quantum computation. So, even though alternative fault-tolerance methods such as completely transversal Clifford+T scheme \cite{2013_PR} and anyonic quantum computing \cite{2003_K} are also being explored, minimization of the number of T gates (or \emph{T-count}) in quantum circuits remain an important and widely studied goal. It has been argued in \cite{2012_F, 2016_AdMVMPS, 2020_dMGM} that it is also important to reduce the maximum number of T gates in any circuit path (or \emph{T-depth}). 
                                
The Solovay-Kitaev algorithm \cite{1997_K, 2006_DN} guarantees that given an $n$-qubit unitary $W$, we can generate a circuit with a ``discrete finite'' universal gate set like Clifford+T, such that the unitary $U$ implemented by the circuit is at most a certain distance from $W$. Here we note that there exists ``infinite continuous''universal gate sets like Clifford+$R_z(\theta)$, with which we can implement any unitary, without any approximation. In this paper we focus on finite universal gate sets, that are more suitable for quantum error correction and fault tolerance. In fact, in quantum computation a set of gates is said to be \emph{universal} if any quantum operation can be approximated to arbitrary accuracy by a quantum circuit involving only those gates \cite{2010_NC}.
A unitary is called \textbf{exactly implementable} by a gate set if there exists a quantum circuit with these gates, that implements it (up to some global phase). Otherwise, it is \textbf{approximately implementable}. Accordingly, a synthesis algorithm can be (a) \emph{exact} when $U=e^{i\phi}W$ ($\phi$ is the global phase); or (b) \emph{approximate} when $d(U,W)\leq\epsilon$ for some $\epsilon> 0$. $d(.)$ is a distance metric. For an unitary $U$ that is exactly implementable by the Clifford+T set, its \textbf{T-count} (denoted by $\tcount(U)$) is the minimum number of T-gates required to implement it, while its \textbf{T-depth} (denoted by $\tcount_d(U)$) is the minimum T-depth of any circuit that implements it. These definitions can be generalized for approximately implementable unitaries and have been described in Section \ref{subsec:TcountDepth}. In this paper we give algorithm for the following two problems.

\paragraph{$\epsilon$-T-COUNT :} Given an $n$-qubit unitary $W$ and $\epsilon\in\real_{\geq 0}$, determine the T-count of a unitary $U$ such that $\tcount(U)\leq\tcount(U')$, where $U,U'$ are $n$-qubit exactly implementable unitaries and $d(U,W),d(U',W)\leq\epsilon$.

\paragraph{$\epsilon$-T-DEPTH :} Given an $n$-qubit unitary $W$ and $\epsilon\in\real_{\geq 0}$, determine the T-depth of a unitary $U$ such that $\tcount_d(U)\leq\tcount_d(U')$, where $U,U'$ are $n$-qubit exactly implementable unitaries and $d(U,W),d(U',W)\leq\epsilon$.

$\tcount(U)$ and $\tcount_d(U)$ are called the \textbf{$\epsilon$-T-count} ($\teps(W)$) and \textbf{$\epsilon$-T-depth} ($\tdeps(W)$) of $W$, respectively. The T-count and T-depth-optimal circuits of $U$ are called \textbf{$\epsilon$-T-count-optimal} and \textbf{$\epsilon$-T-depth-optimal} circuit for $W$, respectively. 
In this paper we use the global phase invariant distance (Section \ref{sec:prelim}) as the distance metric and not the operator norm. This is because the global phase invariant distance ignores the global phase and hence avoids unnecessarily long approximating sequences  that achieve a specific global phase. This can be the reason for the fact that the bound on T-count of single-qubit Z-rotations is less in \cite{2015_KMM}, which works with this distance, compared to \cite{2015_S, 2016_RS}, that work with operator norm. (More discussions can be found in \cite{2021_M}.) This distance is composable \cite{2021_M} and has been used to synthesize unitaries in other models like topological quantum computation \cite{2014_KBS, 2021_JS}.
It is not hard to see that if $\epsilon=0$ then we get the problem of synthesizing T-count and T-depth-optimal circuits for exactly implementable unitaries. In this case, both provable \cite{2014_GKMR, 2020_MM, 2021_GMM} and much more efficient heuristic \cite{2020_MM, 2021_GMM} algorithms have been developed (see Table \ref{tab:compare} for a comparison). We say that an algorithm is \textbf{provable} if its claimed efficiency and correctness or quality of solution (in this case optimality) can be proved by rigorous arguments. An algorithm is \textbf{heuristic} if either one or both of these factors are conjectured to be true.

Any synthesis algorithm will have complexity at least $O(2^n)$, the input size. Placing further optimality constraint makes the problem even harder, in fact impractical to synthesize on a PC after a certain value of $n$. So \textbf{re-synthesis algorithms} have been developed which takes a circuit implementing a unitary and then tries to \emph{reduce} (not \emph{minimize}) a certain resource (see for example \cite{2014_AMM, 2020_GLMM, 2020_HS}). In literature, usually these algorithms do not account for the complexity of synthesizing the circuit and claim a running time of $\poly(n)$. 
A detail study about the relative merit and de-merit of synthesis and re-synthesis algorithms, is beyond the scope of this work. But we would like to point out that the importance of designing better (optimal) synthesis algorithms or studying their complexity cannot be undermined, not only for theoretical reasons but also for the various applications they can have. Apart from guaranteeing optimality, they can be used as sub-routines in compiling large unitaries \cite{2018_HRS, 2020_MSRH, 2021_M}, assess the quality of re-synthesis algorithms, etc. For example, the T-depth-optimal synthesis algorithm of \cite{2021_GMM} was able to generate optimal circuits for standard unitaries like Fredkin, Peres and Quantum OR, which could not be done by the re-synthesis methods used in \cite{2013_AMMR}. In Section \ref{sec:result} we show that we get much less T-count for widely-used multi-qubit unitaries like controlled rotation, compared to the number of T-gates obtained by compiling them first into single-qubit rotations and then replacing the T-count-optimal circuit of each such single-qubit rotation.
 
To the best of our knowledge, before this paper there was no algorithm to determine $\epsilon$-T-count or $\epsilon$-T-depth of arbitrary multi-qubit ($n\geq 1$) unitaries, considering any distance metric. Previous algorithms like \cite{2015_KMM, 2015_S, 2016_RS} synthesize $\epsilon$-T-count-optimal circuits for single qubit Z-rotations. In fact, even if we consider other discrete, finite universal gate sets like Clifford+V or Clifford+CS, there are no algorithms that work for arbitrary multi-qubit unitaries and  minimize the non-Clifford gate count/depth. Even it is not clear how to modify or generalize the methods introduced in these papers. However, our results not only work for multi-qubit unitaries but can also be applied in these alternate bases, as explained in the next section. 

\begin{table}[!htbp]
\centering 
\footnotesize
\begin{tabular}{|c|p{1.5cm}|p{1.5cm}|p{1.1cm}|c|c|c|c|}
 \hline
 Resource & Type of synthesis & Type of \newline algorithm & Distance & Reference & $\#$Qubits ($n$) & Space complexity & Time complexity \\
 \hline\hline
 \multirow{6}{*}{T-count} & \multirow{3}{*}{Exact} & \multirow{2}{*}{Provable} & \multirow{3}{*}{-} & \cite{2014_GKMR} & $>0$ & $O\left(\left(2^n\right)^m\right)$ & $O\left(\left(2^n\right)^m\right)$ \\
 \cline{5-8}
  & & & & \cite{2020_MM} & $>0$ & $O\left(\left(2^n\right)^{2\lceil\frac{m}{c}\rceil}\right)$ & $O\left(\left(2^n\right)^{2(c-1)\lceil\frac{m}{c}\rceil}\right)$  \\
  \cline{3-8}
   & & Heuristic & - & \cite{2020_MM} & $>0$ & $\poly(m,2^n)$ & $\poly(m,2^n)$ \\
   \cline{2-8}
   & \multirow{3}{*}{Approx} & \multirow{3}{*}{Provable} & Op. norm & \cite{2016_RS} & $=1$ ($R_z$) & $\poly(2^n)$ & $\poly(\log(1/\epsilon))$  \\
   \cline{4-8}
   & & & Global Ph. Inv. & \cite{2015_KMM} & $=1$ ($R_z$) & $\poly(2^n)$ & $\exp(1/\epsilon)$  \\
   \cline{4-8}
   & & & Global Ph. Inv. & This work & $>0$ & $O(2^{2n})$ &$O\left(2^{2nm_{\epsilon}+4n}\right)$  \\
   \hline\hline
   \multirow{3}{*}{T-depth} & \multirow{2}{*}{Exact} & Provable & - & \cite{2021_GMM} & $>0$ & $O\left(\left(4^{n^2}\right)^{\lceil\frac{d}{c}\rceil}\right)$ & $O\left(\left(4^{n^2}\right)^{(c-1)\lceil\frac{d}{c}\rceil}\right)$  \\
  \cline{3-8}
   & & Heuristic & - & \cite{2021_GMM} & $>0$ & $\poly\left(d,n2^{5.6n}\right)$ & $\poly\left(d,n2^{5.6n}\right)$  \\
   \cline{2-8}
   & Approx & Provable & Global Ph. Inv. & This work & $>0$ & $O\left(n2^{5.6n}\right)$ & $O\left(2^{2n^2d_{\epsilon}+4n}\right)$ \\
   \hline\hline
   \multirow{2}{*}{Depth} & \multirow{2}{*}{Exact} & \multirow{2}{*}{Provable} & \multirow{2}{*}{-} & \cite{2013_AMMR} & $>0$ & $O\left(\left(2^{kn^2}\right)^{d'/2}\right)$ & $O\left(\left(2^{kn^2}\right)^{d'/2}\right)$ \\
   \cline{5-8}
  & & & & \cite{2021_GMM} & $>0$ & $O\left(\left(2^{kn^2}\right)^{\lceil\frac{d'}{c}\rceil}\right)$ & $O\left(\left(2^{kn^2}\right)^{\lceil(c-1)\frac{d'}{c}\rceil}\right)$ \\
  \hline
\end{tabular}
\caption{Complexity of some state-of-the-art optimal synthesis algorithms. The first column specifies the resouce to be optimized. In the fourth column Op. norm and Global Ph. Inv. denote the operator norm and global phase invariant distance respectively. In the sixth column \cite{2016_RS} and \cite{2015_KMM} give optimal results only for 1-qubit $R_z(\theta)$ gates. In the last two columns $c\geq 2$ and $k$ are constants. $m_{\epsilon}$ and $d_{\epsilon}$ are the $\epsilon$-T-count and $\epsilon$-T-depth of the input approximately implementable unitary respectively. The corresponding terms for exactly implementable unitaries ($\epsilon=0$) are $m$ and $d$ respectively. }
\label{tab:compare}
\end{table}

\subsection{Our contributions}

In this paper we give algorithms that can be used to synthesize (provably) $\epsilon$-T-count and $\epsilon$-T-depth-optimal circuits. We treat arithmetic operations on the entries of a unitary at unit cost and we do not account for the bit-complexity associated with specifying or manipulating them. 
Suppose the input $n$-qubit unitary is $W$, having size $2^n\times 2^n$. Then the space complexity of our algorithms, described in Sections \ref{sec:provable} and \ref{sec:depth}, is $\poly(2^n)$. The time complexity has an exponential dependence on $\teps(W)$ or $\tdeps(W)$ while synthesizing $\epsilon$-T-count or $\epsilon$-T-depth optimal circuit respectively (Table \ref{tab:compare}).

For the design and analysis of our algorithm the following results (Section \ref{sec:provable}) have been crucial. Suppose $E$ is a unitary that is close to $\id$ in the global phase invariant distance i.e. $d(E,\id)\leq\epsilon$. $C_0$ is an $n$-qubit Clifford operator. $EC_0$ \emph{behaves almost like Clifford} $C_0$ i.e. it \emph{approximately inherits} some characteristics from $C_0$. First, expanding both $EC_0$ and $C_0$ in the Pauli basis, we can see that the absolute value of the coefficients (at each point) are \emph{almost equal} (\textbf{amplitude test}). Second, if we expand $(EC_0)P(EC_0)^{\dagger}$ in the Pauli basis then the absolute value of the coefficients is almost $1$ at $P'$ if $C_0PC_0^{\dagger}=P'$, and nearly zero at other points (\textbf{conjugation test}).  
These results may be of independent interest and can be used for resource-optimal synthesis in other bases as described below.

Most discrete universal gate sets consist of Clifford gates and one non-Clifford gate. Consider one such set Clifford+A, where A is a non-Clifford gate and let $U_A$ is a unitary exactly implementable by this set. Since usually the cost of fault-tolerantly implementing the non-Clifford gate is higher, so we are required to optimize the A-count or A-depth, which are defined analogous to T-count and T-depth. One of the tricks in many resource-optimal-synthesis algorithms is to find a nice generating set $\mathcal{G}$ such that it has finite cardinality and $U_A$ can be decomposed as follows.
\begin{eqnarray}
 U_A&=&e^{i\phi}\left(\prod_{i=m}^1G_i\right)C_0\qquad G_i\in\mathcal{G}, C_0\in\cliff_n,\phi\in[0,2\pi), m=\text{A-count/A-depth}
\end{eqnarray}
Each $G_i\in\mathcal{G}$ has a specific property : A-count 1 or A-depth 1. 
Then we design a search algorithm to search for products of $G_i$ such that, we get $U_AC_0^{-1}$ (up to a global phase), or rather $U_A\left(\prod_iG_i\right)^{-1}$ is a Clifford.  We know that for any discrete finite universal gate set not all unitaries are exactly implementable. Let $V_A$ be one such unitary. $U_A$ is a unitary such that $d(U_A,V_A)\leq\epsilon$ and it has the minimum A-count or A-depth among all exactly implementable unitaries within $\epsilon$ distance of $V_A$. Then we can perform amplitude test (Theorem \ref{thm:coeffEC}) and conjugation test (Theorem \ref{thm:conj}, Corollary \ref{cor:conjCliff}) and get an $\epsilon$-A-count-optimal or $\epsilon$-A-depth-optimal decomposition of $V_A$. So it will be interesting and useful to find such nice generating set for other bases, as has been found for Clifford+T \cite{2014_GKMR} (T-count), \cite{2021_GMM} (T-depth) and Clifford+CS \cite{2021_GRT} (CS-count, only for 2-qubit unitaries). One simple way of constructing $\mathcal{G}$ for count-optimality is to write $U_A$ as follows.
\begin{eqnarray}
 U_A&=&e^{i\phi}C_1A_{(q_1)}C_2A_{(q_2)}C_3\ldots C_mA_{(q_m)}C_0   \nonumber \\
 &=&e^{i\phi}\left(C_1A_{(q_1)}C_1^{\dagger}\right)\left(C_1C_2A_{(q_2)}C_2^{\dagger}C_1^{\dagger}\right)\ldots\left(C_1\ldots C_mA_{(q_m)}C_m^{\dagger}\ldots C_1^{\dagger}\right)C_1\ldots C_mC_0 \nonumber \\
 &=&e^{i\phi}\left(C_1A_{(q_1)}C_1^{\dagger}\right)\left(C_2'A_{(q_2)}C_2'^{\dagger}\right)\ldots\left(C_m'A_{(q_m)}C_m'^{\dagger}\right)C_0'   \nonumber
\end{eqnarray}
In the above $A_{(q_i)}$ denotes the A-gate applied on qubit $q_i$. $C_0,\ldots, C_m, C_0'\ldots, C_m'\in\cliff_n$. So $\mathcal{G}=\{CA_{(q_i)}C^{\dagger}:C\in\cliff_n, i\in [n]\}$, which is obviously finite since the Clifford group is finite \cite{1998_CRSS,2008_O}. It is possible to get more compact sets by exploiting other algebraic properties (for example, \cite{2014_GKMR}). A generating set for depth-optimality can be constructed by conjugating products of at least $n$ A-gates on distinct qubits by Clifford, as has been done in \cite{2021_GMM}.

From Table \ref{tab:compare} we see that for exactly implementable unitaries the provable algorithms like \cite{2014_GKMR, 2020_MM, 2021_GMM} had an exponential dependence on T-count and T-depth. Significant improvements were achieved in \cite{2020_MM, 2021_GMM}, where heuristics were designed that led to algorithms with a polynomial dependence on T-count and T-depth. The algorithm in our paper also suffer from exponential dependence on $\epsilon$-T-count and $\epsilon$-T-depth, which usualy have an inverse dependence on $\epsilon$ i.e. $\teps, \tdeps \propto f(1/\epsilon)$. We conjecture that for approximately implementable unitaries it is not possible to have algorithms with a polynomial dependence on these parameters.

\begin{conjecture}
 It is not possible to have $\epsilon$-T-count or $\epsilon$-T-depth-optimal synthesis algorithms with complexity $\poly\left(2^n,\teps\right)$ or $\poly\left(2^n,\tdeps\right)$ respectively.
 \label{conj:approx}
\end{conjecture}

However, from the improvements in T-count obtained by us (Section \ref{sec:result}), we feel it is important to design efficient multi-qubit resource-optimal synthesis algorithms. In many practical quantum algorithms it is not too difficult to decompose a large unitary into smaller ones. We can apply composability rules (for example see \cite{2021_M} for global phase invariant distance) and distribute the errors among these small unitaries, such that the overall error remains within the desired bound. The complexity of resource-optimal synthesis algorithms will determine the maximum size of the component unitaries in a decomposition. The larger the components, the better the resource estimates, as evident from the results in our paper (Section \ref{sec:result}).

\subsection{Related work}

Much work has been done to synthesize a circuit for any multi-qubit unitary (without provable optimality on any resource) \cite{1997_K, 2002_KSVV, 2006_DN, 2010_NC, 2011_F, 2013_KMM2, 2015_S, 2020_dBBVA, 2021_MIC}. 
Comparitively little has been done for arbitrary multi-qubit unitaries, when additional constraints are imposed, like minimizing the T-count or T-depth. To the best of our knowledge, all the previous works for approximately implementable unitaries, have been exclusively for single-qubit unitaries, in fact specifically for $R_z(\theta)$ gates. 
They work with either operator norm \cite{2015_S, 2016_RS} or global phase invariant distance \cite{2015_KMM}.
However, considerable amount of work has been done to synthesize T-count and T-depth-optimal circuits for exactly implementable multi-qubit unitaries. These include algorithms with provable optimality like \cite{2014_GKMR, 2020_MM, 2021_GMM} that employ meet-in-the-middle (MITM) and nested MITM techniques, as well as more efficient heuristic algorithms whose optimality depend on some conjecture \cite{2020_MM, 2021_GMM}. A crisp summary of the complexity of some state-of-the-art optimal synthesis algorithms has been given in Table \ref{tab:compare}. 

Work has also been done to approximate single-qubit unitaries in Clifford+V \cite{2015_R, 2013_BGS, 2015_BBG, 2015_KBRY} basis and perform a CS-count-optimal (exact) synthesis \cite{2021_GRT} of two-qubit unitaries in Clifford+CS basis. 

\subsection{Organization}

Some necessary preliminary definitions and results have been given in Section \ref{sec:prelim}. We describe our algorithm in Section \ref{sec:provable} and provide implementation results in Section \ref{sec:result}. In Section \ref{sec:depth} we discuss about T-depth-optimality.
Finally we conclude in Section \ref{sec:conclude}.

\section{Preliminaries}
\label{sec:prelim}

We write $[K]=\{1,2,\ldots,K\}$. We denote the $n\times n$ identity matrix by $\id_n$ or $\id$ if dimension is clear from the context. We denote the set of $n$-qubit unitaries by $\mathcal{U}_n$. The size of an $n$-qubit unitary is $N\times N$ where $N=2^n$. 
We have given detail description about the n-qubit Pauli operators ($\pauli_n$), Clifford group ($\cliff_n$) and the group ($\clifft_n$) generated by the Clifford and $\T$ gates in Appendix \ref{app:clifford} and \ref{app:Jn}. 
In this section we give some additional definitions and results required for the rest of the paper.

\subsection{The group generated by Clifford and $\T$ gates}

We observe the following when expanding a Clifford in the Pauli basis. 
\begin{fact}[\cite{2010_BS}]
 If $C\in\cliff_n$ then for each $P\in\pauli_n$ $\exists r_P\in\cmplx$, such that $C=\sum_{P\in\pauli_n}r_PP$. Further, if $r_P, r_{P'}\neq 0$ for any pair of $P,P'$, then $|r_P|=|r_{P'}|=r$, for some $r\in\real$.
 \label{fact:cliffCoeff}
\end{fact}

\begin{fact}
Let $Q=\sum_{P\in\pauli_n}q_PP$ be the expansion of a matrix $Q$ in the Pauli basis. Then
$$q_P=\tr\left(QP\right)/N.
$$
Further, if $Q$ is a unitary then
$
    \sum_{P\in\pauli_n}\left|q_P\right|^2=1
$
 \label{fact:trCoeff}
\end{fact}
The proof has been given in Appendix \ref{app:clifford} (Fact \ref{app:fact:trCoeff})
Now consider a Clifford $C$ which has an expansion, as given in Fact \ref{fact:cliffCoeff}. Let there be $M$ ($\leq N^2$) such non-zero coefficients. Since $C$ is a unitary, so we apply Fact \ref{fact:trCoeff} and get the following.
\begin{eqnarray}
 Mr^2=1 \quad\implies r=\frac{1}{\sqrt{M}}\geq \frac{1}{N}.
 \label{eqn:Mr2}
\end{eqnarray}

A unitary $U$ is \textbf{exactly implementable} if there exists a Clifford+T circuit that implements it (up to some global phase), else it is \textbf{approximately implementable}. Specifically, we say $W$ is $\epsilon$-\textbf{approximately implementable} if there exists an exactly implementable unitary $U$ such that $d(U,W)\leq\epsilon$. The Solovay-Kitaev algorithm \cite{1997_K, 2006_DN} guarantees that any unitary is $\epsilon$-approximately implementable, for arbitrary precision $\epsilon\geq 0$.
We denote the set of exactly implementable unitaries by $\clifft_n$. In this paper we use the following distance measure $d(.,.)$, which has been used in previous works like \cite{2011_F, 2015_KMM} (qubit based computing), \cite{2014_KBS, 2021_JS} (topological quantum computing).
\begin{definition}[\textbf{Global phase invariant distance}]
 Given two unitaries $U,W\in\mathcal{U}_n$, we define the global phase invariant distance between them as follows.
 \begin{eqnarray}
  d(U,W)=\sqrt{1-\frac{\left|\tr\left(U^{\dagger}W\right)\right|}{N}}   \nonumber
 \end{eqnarray}
\end{definition}

Composability of this distance with respect to tensor product and multiplication of unitaries, has been derived in \cite{2021_M}. This implies that if $U=\prod_i\left(\otimes_jU_{ij}\right)$, $V=\prod_i\left(\otimes_jV_{ij}\right)$ and $d(U_{ij},V_{ij})\leq\epsilon_{ij}$ then we can upper bound $d(U,V)$ as function of $\epsilon_{ij}$.

\subsection{T-count of circuits and unitaries}
\label{subsec:TcountDepth}

\subsubsection*{T-count of circuits}

The \emph{T-count of a circuit} is the number of T-gates in it. 

\subsubsection*{T-count of exactly implementable unitaries}

The \emph{T-count of an exactly implementable unitary} $U\in\clifft_n$, denoted by $\tcount(U)$, is the minimum number of T-gates required to implement it (up to a global phase). 

In \cite{2014_GKMR} the authors proved the following decomposition result, by which any exactly implementable unitary over the Clifford+T set can be written as a product of T-count 1 unitaries.
\begin{theorem}[Proposition 1 in \cite{2014_GKMR} (re-stated)]
 For any $U\in\clifft_n$ there exists a phase $\phi\in[0,2\pi)$, Cliffords $C_i\in\cliff_n$ and Paulis $P_i\in\pauli_n\setminus\{\id\}$ for $i\in[\tcount(U)]$ such that
\begin{eqnarray}
 U=e^{i\phi} \Big(\prod_{i=\tcount(U)}^{1} R(P_i) \Big)C_0
\end{eqnarray} 
where $R(P_i)=C_iT_{(q_i)}C_i^{\dagger} = \frac{1}{2}(1+e^{\frac{i\pi}{4}})\id+\frac{1}{2}(1-e^{\frac{i\pi}{4}})C_iZ_{(q_i)}C_i^{\dagger} = \frac{1}{2}(1+e^{\frac{i\pi}{4}})\id+\frac{1}{2}(1-e^{\frac{i\pi}{4}})P_i$. 
\label{thm:decompose}
\end{theorem}
Using Fact \ref{fact:cliffConj}, given $P$ and $Z_{(q_i)}$ we can compute (circuit for) $C_i$ efficiently such that $P=C_iZ_{(q_i)}C_i^{\dagger}$. A decomposition of $U$, as in Theorem \ref{thm:decompose}, with the minimum number of T-gates is called a \emph{T-count-optimal decomposition} of $U$.

\subsubsection*{$\epsilon$-T-count of approximately implementable unitaries}
\label{subsubsec:epsTcountDepth}

Let $W\in\mathcal{U}_n$ be an approximately implementable unitary. The \emph{$\epsilon$-T-count} of $W$, denoted by $\teps(W)$, is equal to $\tcount(U)$, the T-count of an exactly implementable unitary $U\in\clifft_n$ such that $d(U,W)\leq\epsilon$ and $\tcount(U)\leq\tcount(U')$ for any $U'\in\clifft_n$ within distance $\epsilon$ of $W$, i.e. $d(U',W)\leq\epsilon$. 
We call a T-count-optimal circuit for any such $U$ as the \emph{$\epsilon$-T-count-optimal} circuit for $W$. 

It is not hard to see that the above definitions are very general and can be applied to any unitary $W\in\mathcal{U}_n$, exactly or approximately implementable. If a unitary is exactly implementable then $\epsilon=0$. In fact, nearly all the following results can be deduced for the special case of exactly implementable unitaries by applying $\epsilon=0$.

\section{An exponential time and polynomial space algorithm}
\label{sec:provable}

In this section we describe an algorithm for determining the $\epsilon$-T-count of an $n$-qubit unitary $W\in\mathcal{U}_n$. This algorithm has space and time complexity $O\left(2^{2n}\right)$ and $O\left(2^{2n\teps(W)+4n}\right)$ respectively. First we derive some results that will help us design our algorithm.

Let $U$ be an exactly synthesizable unitary such that $d(U,W)\leq \epsilon$.
\begin{eqnarray}
 && d(W,U)=\sqrt{1-\frac{|\tr(W^{\dagger}U)|}{N}}\leq\epsilon 
 \implies |\tr(W^{\dagger}U)| \geq N(1-\epsilon^2)
 \label{eqn:dVU}
\end{eqnarray}

Let $U=\left(\prod_{i=t}^1R(P_i) \right)C_0e^{i\phi}$ for some $C_0\in\cliff_n$ and global phase $\phi$. And $W=UE$. Then from Equation \ref{eqn:dVU} we have
\begin{eqnarray}
 \left|\tr\left(E^{\dagger}\right)\right|=\left|\tr\left(E\right)\right|\geq N(1-\epsilon^2).
 \label{eqn:trE}
\end{eqnarray}
The above implies that $d(E,\id)\leq\epsilon$. We have
\begin{eqnarray}
\left|\tr\left(W^{\dagger}\left(\prod_{i=t}^1R(P_i)\right)\right)\right| &=& \left|\tr\left(E^{\dagger}U^{\dagger}\left(\prod_{i=t}^1R(P_i)\right)\right)\right| \nonumber\\
&=& \left|\tr\left(E^{\dagger}e^{-i\phi}C_0^{\dagger}\left(\prod_{i=1}^tR^{\dagger}(P_i)\right)\left(\prod_{i=t}^1R(P_i)\right)\right)\right| \nonumber \\
 &=&\left|\tr\left(E^{\dagger}C_0^{\dagger}\right)\right|=\left|\tr\left(EC_0\right)\right| \nonumber
\end{eqnarray}
and similarly
\begin{eqnarray}
 \left|\tr\left(W^{\dagger}\left(\prod_{i=t}^1R(P_i)\right)P_1\right)\right|=\left|\tr\left(E^{\dagger}C_0^{\dagger}P_1\right)\right|=\left|\tr\left(EC_0P_1\right)\right| \qquad [P_1\in\pauli_n]  \label{eqn:distCliff}
\end{eqnarray}

We now study some properties of $\left|\tr(EC_0P_1)\right|$, which will help us check if we have identified a correct $\prod_{i=t}^1R(P_i)$. For this, we prove the following theorem.
\begin{theorem}
Let $E\in\mathcal{U}_n$ be such that $\left|\tr\left(E\right)\right|\geq N\left(1-\epsilon^2\right)$, for some $\epsilon\geq 0$. $C_0=\sum_{P\in\pauli_n}r_PP$ is an $n$-qubit Clifford. If $\left|\left\{P:r_P\neq 0\right\}\right|=M$ then 
\begin{eqnarray}
\frac{1-\epsilon^2}{\sqrt{M}}-\sqrt{M}\sqrt{2\epsilon^2-\epsilon^4}&\leq&\left|\tr\left(EC_0P_1\right)/N\right|\leq \frac{1}{\sqrt{M}}+\sqrt{M}\sqrt{2\epsilon^2-\epsilon^4}\quad [\text{if } r_{P_1}\neq 0]   \label{eqn:traceRPn0} \\
\text{and}\quad 0&\leq&\left|\tr\left(EC_0P_1\right)/N\right|\leq \sqrt{M}\sqrt{2\epsilon^2-\epsilon^4}\quad [\text{if }r_{P_1}=0] \label{eqn:traceRP0}
\end{eqnarray}
 \label{thm:coeffEC}
\end{theorem}

\begin{proof}
 Since $E$ is unitary, we can expand it in the Pauli basis as 
\begin{eqnarray}
 E = \sum_{P\in\pauli_n} e_P P    
 \label{eqn:Epauli}
\end{eqnarray}
where $e_P=\tr\left(EP\right)/N$ and $\sum_{P}\left|e_P\right|^2=1$ (Fact \ref{fact:trCoeff}). Thus
\begin{eqnarray}
1\geq |e_{\id}| &=& \left|\tr\left(E\right)/N\right| \geq (1-\epsilon^2)\qquad [\text{From inequality \ref{eqn:trE}}]
\label{eqn:eI}  \\
 \text{and }\sum_{P\neq\id} |e_P|^2 &\leq& 1-(1-\epsilon^2)^2 = 2\epsilon^2-\epsilon^4  \label{eqn:eNotI_1} \\
 \implies |e_P|&\leq&\sqrt{2\epsilon^2-\epsilon^4} \qquad [\forall P\neq\id]\label{eqn:absE}
\end{eqnarray}
Since $C_0=\sum_{P\in\pauli_n}r_PP$, from Fact \ref{fact:cliffCoeff} and Equation \ref{eqn:Mr2}, we know that $|r_P|=r=\frac{1}{\sqrt{M}}$ or $r_P=0$. Then
\begin{eqnarray}
 EC_0 &=& \sum_{P\in\pauli_n} r_P EP =\sum_{\substack{P\in\pauli_n\\r_P\neq 0}}r_PEP \nonumber \\
 \text{ and }EC_0P_1&=& r_{P_1}E+\sum_{\substack{P\in\pauli_n\setminus\{P_1\}\\r_P\neq 0}}r_PEP_1' \quad\text{where }P_1\in\pauli_n\setminus\{\id\}\text{ and }PP_1=P_1'\neq\id. \nonumber   
 \end{eqnarray}
 So
 \begin{eqnarray}
 \left|\tr\left(EC_0P_1\right)\right|&=&\left|r_{P_1}\tr\left(E\right)+\sum_{\substack{P\in\pauli_n\setminus\{P_1\}\\r_P\neq 0}}r_P\tr\left(EP_1'\right)\right|   \nonumber \\
 &=&\left|\left(r_{P_1}e_{\id}+\sum_{\substack{P\in\pauli_n\setminus\{P_1\}\\r_P\neq 0}}r_Pe_{P_1'}\right)\right|N\quad [\text{Using Fact \ref{fact:trCoeff}}]  \label{eqn:trace} \\
 &\leq&\left|r_{P_1}\right|\left|e_{\id}\right|N+\sum_{\substack{P\in\pauli_n\setminus\{P_1\}\\r_P\neq 0}}\left|r_P\right|\left|e_{P_1'}\right|N \quad [\because |e_{\id}|\leq 1]   \nonumber \\
 &\leq&\left|r_{P_1}\right|N+\sum_{\substack{P\in\pauli_n\setminus\{P_1\}\\r_P\neq 0}}r\sqrt{2\epsilon^2-\epsilon^4}N   \qquad [\text{Using inequality \ref{eqn:absE}}] \label{eqn:traceUB}
\end{eqnarray}
From equation \ref{eqn:trace} we can also obtain the following lower bound.
\begin{eqnarray}
 \left|\tr\left(EC_0P_1\right)\right|&\geq&\left|r_{P_1}\right|\left|e_{\id}\right|N-\sum_{\substack{P\in\pauli_n\setminus\{P_1\}\\r_P\neq 0}}\left|r_P\right|\left|e_{P_1'}\right|N    \nonumber \\
 &\geq&\left|r_{P_1}\right|(1-\epsilon^2)N-\sum_{\substack{P\in\pauli_n\setminus\{P_1\}\\r_P\neq 0}}r\sqrt{2\epsilon^2-\epsilon^4}N   \quad [\text{Inequality \ref{eqn:absE} and \ref{eqn:eI}}] \label{eqn:traceLB}
\end{eqnarray}
Since $r=\frac{1}{\sqrt{M}}$, we prove the following inequalities. 
\begin{eqnarray}
\frac{1-\epsilon^2}{\sqrt{M}}-\sqrt{M}\sqrt{2\epsilon^2-\epsilon^4}&\leq&\left|\tr\left(EC_0P_1\right)/N\right|\leq \frac{1}{\sqrt{M}}+\sqrt{M}\sqrt{2\epsilon^2-\epsilon^4}\quad [\text{if } r_{P_1}\neq 0]   \nonumber \\
\text{and}\quad 0&\leq&\left|\tr\left(EC_0P_1\right)/N\right|\leq \sqrt{M}\sqrt{2\epsilon^2-\epsilon^4}\quad [\text{if }r_{P_1}=0] \nonumber
\end{eqnarray}
\end{proof}

Basically this theorem says that if $E$ is close to identity then distribution of absolute-value-coefficients of $EC_0$ and $C_0$ in the Pauli basis expansion, is almost similar. In fact, we can have a more general theorem that can be deduced from the calculations in Theorem \ref{thm:coeffEC}.
\begin{theorem}
Let $E\in\mathcal{U}_n$ be such that $\left|\tr\left(E\right)\right|\geq N\left(1-\epsilon^2\right)$, for some $\epsilon\geq 0$. $Q=\sum_{P\in\pauli_n}q_PP$ is an $n$-qubit unitary. Then for each $P_1\in\pauli_n$,
\begin{eqnarray}
(1-\epsilon^2)|q_{P_1}|-\sum_{P\in\pauli_n\setminus\{P_1\}}|q_P|\sqrt{2\epsilon^2-\epsilon^4} &\leq&\left|\tr\left(EQP_1\right)/N\right|\leq |q_{P_1}|+\sum_{P\in\pauli_n\setminus\{P_1\}}|q_P|\sqrt{2\epsilon^2-\epsilon^4}. \nonumber
\end{eqnarray}

 \label{thm:coeffEQ}
\end{theorem}


So we can define two sets $\mathcal{S}_1$ and $\mathcal{S}_0$ as follows.
\begin{eqnarray}
 \mathcal{S}_1&=&\left\{\left|\tr\left(EC_0P_1\right)/N\right|:r_{P_1}\neq 0\right\}    \label{eqn:S1} \\
 \mathcal{S}_0&=&\left\{\left|\tr\left(EC_0P_1\right)/N\right|:r_{P_1}= 0\right\}    \label{eqn:S0} 
\end{eqnarray}
From our results so far, it follows that for small enough $\epsilon$ (which is usually the case in nearly all applications) the values in $\mathcal{S}_1$ are nearly equal, while those in $\mathcal{S}_0$ are nearly $0$. Let $\Delta=\max_{t_1\in\mathcal{S}_1,t_0\in\mathcal{S}_0} (t_1-t_0)$. Then to get a positive difference we have the following.
\begin{eqnarray}
 \Delta&\geq& \frac{1-\epsilon^2}{\sqrt{M}}-2\sqrt{M(2\epsilon^2-\epsilon^4)}>0 \nonumber \\
 \implies &&\epsilon^4-2\epsilon^2+\frac{1}{1+4M^2}\geq 0   \nonumber
\end{eqnarray}
Solving this we obtain the following conditions.
\begin{eqnarray}
 \epsilon^2\geq 1+\sqrt{1-\frac{1}{1+4M^2}}\quad \text{or}\quad\epsilon^2\leq 1-\sqrt{1-\frac{1}{1+4M^2}}
\end{eqnarray}
Since usually $\epsilon<1$, so we consider the second inequality. Expanding the term in the square root we obtain
\begin{eqnarray}
 \epsilon^2\leq \frac{1}{2}\cdot\frac{1}{1+4M^2}+\frac{1}{2!}\cdot\frac{1}{2\cdot 2}\left(\frac{1}{1+4M^2}\right)^2+\frac{1}{3!}\cdot\frac{3}{2\cdot 2\cdot 2}\left(\frac{1}{1+4M^2}\right)^3+\ldots    \nonumber
\end{eqnarray}
Since this function decreases with $M$ and $1\leq M\leq N^2$, so we can say that
\begin{eqnarray}
 \epsilon^2\leq 0.105572809\quad\implies \epsilon\leq 0.3249196962.
\end{eqnarray}
For all practical purposes, the value of $\epsilon$ is much smaller than this. So we can easily distinguish the sets $\mathcal{S}_0$ and $\mathcal{S}_1$.

\subsection{Algorithm}
\label{subsec:algo}

Now we are in a position to describe our exhaustive search algorithm, $\mathcal{A}_{MIN}$ (Algorithm \ref{alg:min}), that determines the $\epsilon$-T-count of a unitary $W\in\mathcal{U}_n$. This is an iterative procedure, where in every iteration we decide  if $\teps(W)=m$ for increasing values of a variable $m$. 

\begin{algorithm}
\footnotesize
 \caption{$\mathcal{A}_{MIN}$}
 \label{alg:min}
 
 \KwIn{(i) $W\in \mathcal{U}_n$, (ii) $\epsilon\geq 0$.}
 
 \KwOut{ $\teps(W)$.}
 
 $m\leftarrow 1$, $decision\leftarrow\text{NO}$    \;
 \While{(1)}
 {
    $decision\leftarrow\mathcal{A}_{DECIDE}(W,m,\epsilon)$ \;
    \eIf{$decision==\text{YES}$}
    {
        \textbf{return} $m$ \;
    }
    {
        $m\leftarrow m+1$   \;
    }
 }
 
 \end{algorithm}

 \begin{algorithm}
\footnotesize
 \caption{$\mathcal{A}_{DECIDE}$}
 \label{alg:decide}
 
 \KwIn{(i) $W\in \mathcal{U}_n$, (ii) integer $m > 0$, (iii) $\epsilon\geq 0$.}
 
 \KwOut{ $\text{YES}$ if $\exists U\in\clifft_n$ such that $d(U,W)\leq\epsilon$ and $\tcount(U)\leq m$; else $\text{NO}$.}
 
 $\mathcal{S}\leftarrow \{R(P):P\neq\id\}$ \label{Adec:S}\;
 
 \For{every $\widetilde{U}=\prod_{i=m}^1U_j$ such that $U_i\in\mathcal{S}$ and $U_i\neq U_{i+1}$}
 {
    $W'=W^{\dagger}\widetilde{U}$    \label{decide:W}\;
    $\mathcal{S}_c\leftarrow\{\left|\tr\left(W'P\right)/N\right|:P\in\pauli_n\}$ and sort this set in descending order \label{decide:Sc}\;
    \For{$M=1,2\ldots N^2$}
    {
        $\mathcal{S}_1\leftarrow $ First $M$ terms in $\mathcal{S}_c$  \label{decide:S1}\;
        $\mathcal{S}_0=\mathcal{S}_c\setminus\mathcal{S}_1$ \label{decide:S0}\;
        \If{each term in $\mathcal{S}_1 \in \left[\frac{1-\epsilon^2}{\sqrt{M}}-\sqrt{M(2\epsilon^2-\epsilon^4)},\frac{1}{\sqrt{M}}+\sqrt{M(2\epsilon^2-\epsilon^4)}\right]$ and each term in $\mathcal{S}_0\in\left[0,\sqrt{M(2\epsilon^2-\epsilon^4)}\right]$}
        {
            \If{YES $\leftarrow \mathcal{A}_{CONJ}(W',\epsilon)$}
            {
            \textbf{return} $\text{YES}$ \;
            }
        }
    }
 }
 \textbf{return} $\text{NO}$\;
\end{algorithm}

 \begin{algorithm}
\footnotesize
 \caption{$\mathcal{A}_{CONJ}$}
 \label{alg:conj}
 
 \KwIn{(i) $W'\in \mathcal{U}_n$, (ii) $\epsilon\geq 0$.}
 
 \KwOut{ $\text{YES}$ if $\exists C_0\in\cliff_n, E\in\mathcal{U}_n$ such that $W'=E^{\dagger}C_0$, where $d(E,\id)\leq\epsilon$ ; else $\text{NO}$.}
 
 $p\leftarrow 1$    \;
 
 \For{every $P_{out}\in \pauli_n$}
 {
    \If{$p==1$}
    {
        $p\leftarrow 0$ \;
    }

    \For{every $P_{in}\in\pauli_n$}
    {
        \If{$(1-4\epsilon^2+2\epsilon^4) \leq\left|\tr\left(W'P_{out}W'^{\dagger}P_{in}\right)\right|/N\leq 1$}
        {
            $p\leftarrow p+1$   \;
            \If{$p>1$}
             {
                \textbf{return} NO  \;
            }
        }
        \If{$2\epsilon <\left|\tr\left(W'P_{out}W'^{\dagger}P_{in}\right)\right|/N < (1-4\epsilon^2+2\epsilon^4) $}
        {
            \textbf{return} NO \;
        }
    }
 }
 
 \textbf{return} YES    \;
 \end{algorithm}
 
The main idea to solve the decision version is as follows. Suppose we have to test if $\teps(W)=m$. From the definitions given in Section \ref{subsubsec:epsTcountDepth}, we know that if this is true then $\exists U\in\clifft_n$ such that $\tcount(U)=m$. Let $U=\left(\prod_{i=m}^1U_i\right)C_0$ where $C_0\in\cliff_n$ and $U_i\in\{R(P):P\neq\id\}$. Let $\widetilde{U}=\prod_{i=m}^1U_i$. To test if we have guessed the correct $\widetilde{U}$, we can apply the results deduced in the previous section. Specifically we calculate $W'=W^{\dagger}\widetilde{U}$ and then calculate the set $\mathcal{S}_c=\left\{|\tr(W'P)/N|:P\in\pauli_n\right\}$, of coefficients. Then we check if we can distinguish two subsets $\mathcal{S}_0$ and $\mathcal{S}_1$, as shown in Equation \ref{eqn:S1} and \ref{eqn:S0}, for some $1\leq M\leq N^2$. Further details have been given in Algorithm \ref{alg:decide}.
Let us call this the \textbf{amplitude test}. After passing this test we have a unitary of the form $E^{\dagger}Q$ where $Q$ is a unitary. This test sort of \emph{filters out} the approximate values of the coefficients of $Q$ in the Pauli basis (Theorem \ref{thm:coeffEQ}). So after passing this test $Q$ will be a unitary with \emph{equal} or \emph{nearly equal} amplitudes or coefficients (absolute value) at some points and zero or nearly zero at other points.
To ensure $Q$ is a Clifford i.e. $W'=E^{\dagger}C_0$ for some $C_0\in\cliff_n$, we perform the \textbf{conjugation test} (Algorithm \ref{alg:conj}) for further verification. 

\begin{theorem}
Let $E,Q\in\mathcal{U}_n$ such that $d(E,\id)\leq\epsilon$. $P'\in\pauli_n\setminus\{\id\}$ such that $QP'Q^{\dagger}=\sum_{P}\alpha_{P}P$, where $\alpha_{P}\in\cmplx$. Then for each $P''\in\pauli_n$,
\begin{eqnarray} 
\min\{0,|\alpha_{P''}|(1-4\epsilon^2+2\epsilon^4) 
-2\epsilon\sum_{P\neq P''}|\alpha_P| \}
&\leq& \left|\tr\left((E^{\dagger}QP'Q^{\dagger}E)P''\right)\right|/N \nonumber \\
&\leq& \max\{ |\alpha_{P''}|+2\epsilon\sum_{P\neq P''}|\alpha_P|, 1\} 
\nonumber   
\end{eqnarray}

\label{thm:conj}
\end{theorem}

\begin{proof}
 We have the following.
 \begin{eqnarray}
 E^{\dagger}(QP'Q^{\dagger})E&=&\sum_{\hat{P}}|e_{\hat{P}}|^2\hat{P}\left(\sum_{P}\alpha_{P}P\right)\hat{P}+\sum_{\hat{P}\neq\tilde{P}}\conj{e_{\hat{P}}}e_{\tilde{P}}\hat{P}\left(\sum_{P}\alpha_{P}P\right)\tilde{P}    \nonumber \\
 &=&|e_{\id}|^2\sum_{P}\alpha_{P}P+\sum_{\hat{P}\neq\id,P}|e_{\hat{P}}|^2\alpha_{P}\hat{P}P\hat{P}+\sum_{\hat{P}\neq\tilde{P},P}\conj{e_{\hat{P}}}e_{\tilde{P}}\alpha_{P}\hat{P}P\tilde{P}  \label{eqn:conj0}
\end{eqnarray}
Multiplication by $P''$ gives us the following.
\begin{eqnarray}
 (E^{\dagger}QP'Q^{\dagger}E)P''&=&|e_{\id}|^2\alpha_{P''}\id+|e_{\id}|^2\sum_{P\neq P''}\alpha_{P}PP''+\sum_{\hat{P}\neq\id}|e_{\hat{P}}|^2\alpha_{P''}(\pm\id)   \nonumber \\
 &&+\sum_{\hat{P}\neq\id,P\neq P''}|e_{\hat{P}}|^2\alpha_{P}\hat{P}P\hat{P}P''+\sum_{\hat{P}\neq\tilde{P}}\conj{e_{\hat{P}}}e_{\tilde{P}}\alpha_{P''}\hat{P}P''\tilde{P}P'' \nonumber \\
 &&+\sum_{\hat{P}\neq\tilde{P},P\neq P''}\conj{e_{\hat{P}}}e_{\tilde{P}}\alpha_P\hat{P}P\tilde{P}P'' \nonumber
\end{eqnarray}
So,
\begin{eqnarray}
 \left|\tr\left((E^{\dagger}QP'Q^{\dagger}E)P''\right)/N\right| &\leq&|\alpha_{P''}|\sum_{\hat{P}}|e_{\hat{P}}|^2
 +\sum_{\hat{P}\neq\tilde{P},P\neq P''}  |\conj{e_{\hat{P}}}e_{\tilde{P}}\alpha_P| \left|\tr(\hat{P}P\tilde{P}P'')/N\right|   \nonumber\\
&=&|\alpha_{P''}| +\sum_{P\neq P''}|\alpha_P|\sum_{\hat{P}\neq\tilde{P}} |\conj{e_{\hat{P}}}e_{\tilde{P}}| \left|\tr(\hat{P}P\tilde{P}P'')/N\right|  
\qquad [\text{Fact }\ref{fact:trCoeff}]\nonumber
\end{eqnarray}
Given $P'', P, \hat{P}$, we can have $\hat{P}P\tilde{P}P''=\pm\id$ for one particular value of $\tilde{P}$. 
\begin{eqnarray}
 \left|\tr\left((E^{\dagger}QP'Q^{\dagger}E)P''\right)/N\right| \leq|\alpha_{P''}|+ \sum_{P\neq P''} |\alpha_P|\sum_{\hat{P}} |\conj{e_{\hat{P}}}||e_{\hat{P'}}| \qquad [\hat{P'}\neq \hat{P}\text{ is such that } \hat{P}P\hat{P'}P''=\pm\id]
  \nonumber
\end{eqnarray}
 Let $\id P \hat{P_0'}P''=\pm\id$ and $\hat{P_0}P\id P''=\pm\id$, for some Paulis $\hat{P_0'},\hat{P_0}\in\pauli_n\setminus\{\id\}$. Then we can write
\begin{eqnarray}
 \left|\tr\left((E^{\dagger}QP'Q^{\dagger}E)P''\right)/N\right| &\leq&|\alpha_{P''}|+ \sum_{P\neq P''} |\alpha_P|\left(|e_{\id}||e_{\hat{P_0'}}|+|e_{\hat{P_0}}||e_{\id}|+
 \sum_{\hat{P}\neq\id,\hat{P_0}} |e_{\hat{P}}||e_{\hat{P'}}| \right) \nonumber \\
 &\leq&|\alpha_{P''}|+ \sum_{P\neq P''} |\alpha_P|\left(|e_{\hat{P_0'}}|+|e_{\hat{P_0}}|+
 \sum_{\hat{P}\neq\id,\hat{P_0}} |e_{\hat{P}}||e_{\hat{P'}}| \right) \quad [\text{Equation \ref{eqn:eI}}]
  \nonumber
\end{eqnarray}
In Appendix \ref{app:bound} we show that $\sum_{\hat{P}\neq\id,\hat{P_0}}|e_{\hat{P}}||e_{\hat{P'}}|\leq (2\epsilon^2-\epsilon^4)$. We observe that in the above inequality we have taken $|e_{\id}|\leq 1$, but if $|e_{\id}|=1$ then $|e_P|=0$ for any $P\neq\id$, since $\sum_{P}|e_P|^2=1$. To get non-zero values for the sum within bracket $|e_{\id}|<1$. If we have to maximize $|e_{\hat{P_0'}}|+|e_{\hat{P_0}}|$ given Equation \ref{eqn:eNotI_1}, then if we consider an optimization problem with these two variables only, then it is not difficult to see that the maximum occurs if they have the same value. That is $|e_{\hat{P_0}}|+|e_{\hat{P_0'}}|\leq 2\sqrt{\frac{2\epsilon^2-\epsilon^4}{2}}\approx 2\epsilon$. Ignoring higher order terms of $\epsilon$, we can write the following.
\begin{eqnarray}
 \left|\tr\left((E^{\dagger}QP'Q^{\dagger}E)P''\right)/N\right| \leq|\alpha_{P''}|+2\epsilon \sum_{P\neq P''} |\alpha_P|
  \nonumber
\end{eqnarray}
We also have the following lower bound using similar reasoning as above.
\begin{eqnarray}
 \left|\tr\left((E^{\dagger}QP'Q^{\dagger}E)P''\right)/N\right| &\geq&\nonumber|\alpha_{P''}|\left(|e_{\id}|^2-\sum_{\hat{P}\neq\id}|e_{\hat{P}}|^2\right)-\sum_{P\neq P''}|\alpha_P|\sum_{\hat{P}\neq\tilde{P}}|\conj{e_{\hat{P}}}e_{\tilde{P}}|\left|\tr(\hat{P}P\tilde{P}P'')/N\right|  \nonumber \\
 &\geq&|\alpha_{P''}|(1-4\epsilon^2+2\epsilon^4)-2\epsilon\sum_{P\neq P''}|\alpha_P| \nonumber 
\end{eqnarray}

\end{proof}

And we have the following corollary.
\begin{corollary}

Let $C_0\in\cliff_n$ and $P'\in\pauli_n$ such that $C_0P'C_0^{\dagger}=\tilde{P}\in\pauli_n$. If $E\in\mathcal{U}_n$ such that $d(E,\id)\leq\epsilon$, then
\begin{eqnarray}
 (1-4\epsilon^2+2\epsilon^4) &\leq&\left|\tr\left((E^{\dagger}C_0P'C_0^{\dagger}E)P''\right)\right|/N\leq 1 \qquad \text{if }  P''=\tilde{P}  \nonumber   \\
0&\leq&\left|\tr\left((E^{\dagger}C_0P'C_0^{\dagger}E)P''\right)\right|/N\leq  2\epsilon  \qquad \text{else. }   \nonumber
\end{eqnarray}

 \label{cor:conjCliff}
\end{corollary}

The above theorem and corollary basically says that $EC_0$ (or $E^{\dagger}C_0$) \emph{approximately} inherits the conjugation property of $C_0$. For each $P'\in\pauli_n$, if we expand $C_0P'C_0^{\dagger}$ in the Pauli basis then the absolute value of the coefficients has value 1 at one point, 0 in the rest. If we expand $EC_0P'C_0^{\dagger}E^{\dagger}$ in the Pauli basis then one of the coefficients (absolute value) will be almost 1 and the rest will be almost 0. From Theorem \ref{thm:conj} this pattern will not show for at least one Pauli $P'''\in\pauli_n$ if we have $E^{\dagger}Q$, where $Q\notin\cliff_n$. If we expand $EQP'''Q^{\dagger}E^{\dagger}$ or $E^{\dagger}QP'''Q^{\dagger}E$ in the Pauli basis then the \emph{spike in the amplitudes} will be in at least two points. Also, we observe that $2\epsilon < (1-4\epsilon^2+2\epsilon^4)$ for any $\epsilon\leq 0.31$. Thus there exists a distinguishable gap between the two cases of Corollary \ref{cor:conjCliff}. For all practical purposes $\epsilon$ is much less than this value.

\subsection{Synthesizing T-count-optimal circuits}
\label{subsec:cktSynth}

So far we have been able to determine $\teps(W)$ for any $W\in\mathcal{U}_n$. We now describe how we can synthesize $\epsilon$-T-count-optimal circuit for $W$, using the above algorithms. It is easy to see that $\mathcal{A}_{DECIDE}$ can return a sequence $\{U_m,\ldots,U_1\}$ of unitaries such that $U=\left(\prod_{i=\tcount(U)}^1U_i\right)C_0e^{i\phi}$ (for some $C_0\in\cliff_n$) and $\tcount(U)=\teps(W)$. We can efficiently construct circuits for each $U_i\in\{ R(P):P\neq\id\}$ using Fact \ref{fact:cliffConj}. So what remains, is to determine $C_0$. Then we can efficiently construct a circuit for it, for example, by using results in \cite{2004_AG}.

If $W=UE$ then at step \ref{decide:W} of Algorithm \ref{alg:decide} we calculate $W'=W^{\dagger}\widetilde{U}=e^{-i\phi}E^{\dagger}C_0^{\dagger}$, where $\widetilde{U}=\prod_{i=m}^1U_i$. 
From Algorithm \ref{alg:decide} we can also obtain the following information : (1) set $\mathcal{S}_1$, as defined in Equation \ref{eqn:S1}, (2) $M=|\mathcal{S}_1|$. Thus we can calculate $r=\frac{1}{\sqrt{M}}$ and from step \ref{decide:Sc} we can actucally calculate the set $\widetilde{\mathcal{S}_1}=\left\{(t_P,P): |t_P|=|\tr(W'P)/N|\in\mathcal{S}_1\right\}$.
From Equation \ref{eqn:trace} we can say that for small enough $\epsilon$ (say $<<\frac{1}{M}$) we have
\begin{eqnarray}
 \frac{\tr\left(E^{\dagger}C_0^{\dagger}P_1\right)}{\tr\left(E^{\dagger}C_0^{\dagger}P_2\right)}&\approx&\frac{\overline{r_{P_1}}}{\overline{r_{P_2}}}
\end{eqnarray}
We perform the following steps.
\begin{enumerate}
\item Calculate $a_P=\frac{t_P}{t_{\id}}=\frac{\tr(W'P)}{\tr(W')}$, where $(t_P,P)\in\widetilde{\mathcal{S}_1}$ (or equivalently $|t_P|\in\mathcal{S}_1$). We must remember that $(t_{\id},\id)\in\widetilde{\mathcal{S}_1}$.  
 
 We explained that from Equation \ref{eqn:trace}, $a_P\approx\frac{\overline{r_P}}{\overline{r_{\id}}}$. From Fact \ref{fact:cliffCoeff} we know that $\frac{|\overline{r_P}|}{|\overline{r_{\id}}|}=\frac{|r_P|}{|r_{\id}|}=1$. So we adjust the fractions such that their absolute value is $1$. For small enough $\epsilon$ this adjustment is not much and so with a slight abuse we use the same notation for the adjusted values.
 
 \item Select $c,d\in\real$ such that $c^2+d^2=r^2$. Let $\widetilde{r_{\id}}=c+di$. Then we claim that the Clifford $\widetilde{C_0}=\widetilde{r_{\id}}\sum_{P:r_P\neq 0}\overline{a_P}P$ is sufficient for our purpose. 
\end{enumerate}

It is not hard to see that $\widetilde{C_0}=e^{i\phi'}C_0$ for some $\phi'\in [0,2\pi)$. Thus if $U'=\left(\prod_{i=m}^1U_i\right)\widetilde{C_0}$, then $\tcount(U')=\tcount(U)$ and $d(U',W)\leq\epsilon$. 

\subsection{Complexity}

We first do a worst case analysis of both space and time complexity and then mention some practical considerations.

\subsection*{Time complexity : worst case analysis}

We first analyse the time complexity of $\mathcal{A}_{CONJ}$. The outer and inner loop at steps 2 and 6, respectively, can run at most $N^2* N^2= N^4$ times, where $N=2^n$. At step 7 multiplication of four $N\times N$ matrices take $O(N^2)$ time and calculating trace takes $N$ time steps. So overall complexity at step 7 is dominated by $O(N^2)$. We note that at step 13 we do not need to calculate the product and trace again. In the worst case every loop at steps 2 and 6 are implemented, incurring an overall time complexity of $O(N^6)$.

Now we analyse the time complexity of $\mathcal{A}_{DECIDE}$, when testing for a particular T-count $m$. The algorithm loops over all possible products of $m$ unitaries $U_j$, which is $R(P)$ in case of T-count-decision. Since there can be $N^2-1$  non-identity Paulis $P$, so this loop happens at most $N^{2m}$ times. Now in each such loop we do $m$ matrix multiplications at step 2 and 3. This has time complexity $O(mN^2)$. At step 4 we make a list of $N^2$ real numbers. Each is obtained by multiplying two $N\times N$ matrices and then taking trace. So time complexity for making this list is $ O(N^4)$. Sorting this list takes time $O(nN^2)$. The inner loop 5-13 happens $N^2$ times. Each of the list elements is checked and so step 8 has complexity $O(N^2)$. Now let within the inner loop the conjugation test is called $k_1$ times. So the loop 5-13 incurs a complexity $O(k_1\cdot N^6+(N^2-k_1)N^2)$, when $k_1>0$, else it is $O(N^4)$. Let $k$ is the number of outer loops (steps 2-14) for which conjugation test is invoked in the inner loop 5-13 and $k_1$ is the maximum number of times this test is called within any 5-13 inner loop.
Then the overall complexity of $\mathcal{A}_{DECIDE}$ is $O((N^{2m}-k)\cdot (mN^2+N^4+nN^2+N^4)+k\cdot (mN^2+N^4+nN^2+k_1N^6+(N^2-k_1)N^2 ) )\in O((N^{2m}-k)\cdot N^4+k\cdot (k_1N^6+(N^2-k_1)N^2 ))$, assuming $m<N^2$. 

The conjugation test is invoked only if a unitary passes the amplitude test. We assume that the occurrence of non-Clifford unitaries with equal amplitude is not so frequent such that $kk_1< N^{2m}-k$. (We did check this in our implementations.) Thus $\mathcal{A}_{DECIDE}$ has a complexity of $O(N^{2m+4})$, for one particular value of $m$. 
Hence, the overall algorithm $\mathcal{A}_{MIN}$ has a time complexity $O(N^{2\teps(W)+4})\in O(2^{2n\teps(W)+4n})$, with the given assumption.

\paragraph{Time complexity - practical considerations :} 
\begin{enumerate}
 \item It is not hard to see that if $[P_1,P_2]=P_1P_2-P_2P_1=0$ then $R(P_1)R(P_2)=(\alpha\id+\beta P_1)(\alpha\id+\beta P_2)=R(P_2)R(P_1)$, where $\alpha=\frac{1}{2}\left(1+e^{i\pi/4}\right)$ and $\beta=\frac{1}{2}\left(1-e^{i\pi/4}\right)$. Thus $$\left(\prod_iR(P_i)\right)R(P_1)R(P_2)\left(\prod_jR(P_j)\right)=\left(\prod_iR(P_i)\right)R(P_2)R(P_1)\left(\prod_jR(P_j)\right).$$ So in step 2 of Algorithm \ref{alg:decide} we need not loop over all $m$-length products of $R(P)$. It is easy to check if two $n$-qubit Paulis commute. There are even number of places where the respective 1-qubit Paulis are non-identity and different. We need not go into actual matrix multiplications. This can speed-up the actual running time by orders of magnitude. For example, for the unitaries considered by us, we obtained a speed-up of 5-10 times. In fact, it may be possible to show that the asymptotic complexity also decreases. One can work out more such symmetries in order to prune the search space.
 
 \item We already made an assumption that the number of non-Cliffords that pass the amplitude test is much less. Even if such a unitary is tested in $\mathcal{A}_{CONJ}$, we need not loop over $N^4$ times. As soon as there are 2 spikes for any outer loop Pauli $P_{out}$, the program exits (step 10). If a non-spike is ``far enough'' from 0 then also the program exits (step 14). So in most cases testing a non-Clifford with equal amplitudes take less time. If there is Clifford then all the $N^4$ loops have to run, but then it implies that $\mathcal{A}_{DECIDE}$ has obtained T-count.
 
 \item Most of the matrix multiplications, especially by Paulis are sparse, so here the run-time complexity is less. In step 2 of Algorithm \ref{alg:decide} one has to repeatedly multiply a unitary $U$ by $R(P)=\alpha\id+\beta P$. Since $P$ is sparse, we can first multiply $U$ by $P$,  then multiply each non-zero off-diagonal element by $\beta$ and finally add $\alpha$ to the diagonal. This can reduce some practical running time.
\end{enumerate}

\subsection*{Space complexity}

The input to our algorithm is a $N\times N$ unitary, with space complexity $O(N^2)$. In step 1 of $\mathcal{A}_{DECIDE}$, we can store the single qubit Paulis and calculate $R(P)$ whenever required. We require $O(N^2)$ space to perform matrix multiplication of two $N\times N$ matrices. In $\mathcal{A}_{DECIDE}$, we either store a $N\times N$ matrix or a list of $N^2$ real numbers (step 4). Even in $\mathcal{A}_{CONJ}$ we store one $N\times N$ matrix. Hence the overall space complexity is $O(N^2)\in O(2^{2n})$, without storing $R(P)$. This increases running time because we have to calculate the $n$-qubit Paulis and $R(P)$ repeatedly. But the asymptotic time complexity remains unchanged. 

If we store the $n$-qubit Paulis or $R(P)$, then we require $O(N^4)$ space. This factor dominates and overall space complexity is $O(N^4)\in O(2^{4n})$. In this approach the actual running time reduces but the asymptotic time complexity remains same.

\section{Implementation results}
\label{sec:result}

We implemented our algorithm $\mathcal{A}_{MIN}$ in standard C++17 on an Intel(R) Core(TM) i7-7700K CPU at 4.2GHz, with 8 cores and 48 GB RAM, running FreeBSD 13.1. We compiled the code using clang++ 13.0.0. We used OpenMP~\cite{OpenMP} for parallelization  and the Eigen~3 matrix library~\cite{Eigen} for some of the matrix operations. We applied our algorithm to return the T-count-optimal decomposition of the following unitaries : (i) 1-qubit $R_z(\theta)$ and $R_k$; (ii) 2-qubit controlled-$R_z(\theta)$ ($cR_z(\theta)$); (iii) 2-qubit controlled $R_k$ ($cR_k$); (iv) 3-qubit double controlled-$R_z(\theta)$ ($ccR_z(\theta)$); (v) 3-qubit double controlled-$R_k$ ($ccR_k$); (vi) 2-qubit Givens rotation ($Givens(\theta)$); (vii) 2-qubit QFT.
\begin{eqnarray}
\tiny
&(i)& R_z(\theta)=\begin{bmatrix}
             e^{-i\theta/2} & 0 \\
             0 & e^{i\theta/2}
            \end{bmatrix}; \qquad   
R_k=\begin{bmatrix}
             1 & 0 \\
             0 & e^{2\pi i/2^k}
            \end{bmatrix}  =e^{2\pi i/2^{k+1}}R_z\left(\frac{2\pi}{2^k}\right)  \nonumber \\
&(ii)& cR_z(\theta)=\diag\left(1,1, e^{-i\theta/2}, e^{i\theta/2}\right) \qquad
   (iii)\quad cR_k=\diag\left(1,1,1, e^{2\pi i/2^k}\right) \nonumber \\
&(iv)& ccR_z(\theta)=\diag\left(1,1,1,1,1,1,e^{-i\theta/2}, e^{i\theta/2} \right)
\nonumber \\
&(v)& ccR_k=\diag\left(1,1,1,1,1,1,1, e^{2\pi i/2^k} \right) 
\label{eqn:matrices} \qquad
(vi)\quad Givens(\theta)=\begin{bmatrix}
                       1 & 0 & 0 & 0 \\
                       0 & \cos\theta & -\sin\theta & 0 \\
                       0 & \sin\theta & \cos\theta & 0 \\
                       0 & 0 & 0 & 1
                      \end{bmatrix} \nonumber
\end{eqnarray}
For convenience, we have denoted some diagonal matrices by $\diag(a,b,c,\ldots)$ which implies that the matrix has elements $a,b,c,\ldots$ along the diagonal and $0$ elsewhere. We used Quantum++ \cite{2018_G} to obtain the unitary from the circuit given in \cite{2010_NC} (Figure \ref{fig:qft} in Appendix \ref{app:Tcount}). 

Controlled-$R_z(\theta)$ gates appear in many important quantum algorithms like Quantum Fourier Transform (QFT), phase estimation, factorization, order finding, hidden subgroup problem, Grover's search \cite{2010_NC}, quantum simulations \cite{2012_JWMetal}. Givens rotation appears in a number of quantum chemistry aplications \cite{2021_AdMQetal}. Before our work there was no algorithm for T-count of multi-qubit approximately implementable unitaries. In fact they returned optimal results only for single qubit $R_z(\theta)$. If possible other unitaries were decomposed into a sequence of $R_z(\theta)$ gates, for which we know the following empirical formula from \cite{2015_KMM}, where the T-count is averaged over $\theta$.
\begin{eqnarray}
 \tcount^{\epsilon}(R_z(\theta)) &=& 3.067\log_2\left(1/\epsilon\right)-4.322
 \label{eqn:KMM15}
\end{eqnarray}
\begin{figure}[h]
\centering
\begin{subfigure}{0.24\textwidth}
 \centering
 \includegraphics[width=2.5cm, height=1.5cm]{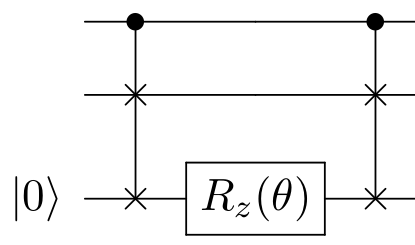}
 \caption{}
 \label{fig:cRz1}
\end{subfigure}
\hfil
\begin{subfigure}{0.24\textwidth}
 \centering
 \includegraphics[width=3cm, height=1.5cm]{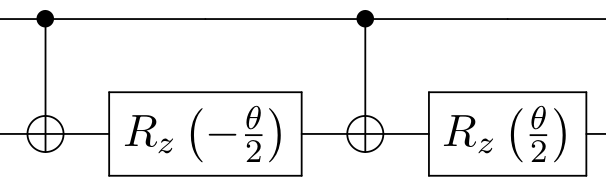}
 \caption{}
 \label{fig:cRz2}
 \end{subfigure}
 \hfil
\begin{subfigure}{0.24\textwidth}
 \centering
 \includegraphics[width=2.5cm, height=1.5cm]{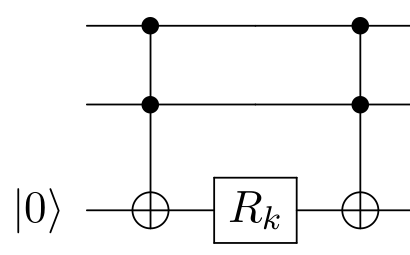}
 \caption{}
 \label{fig:cRk}
\end{subfigure}
\hfil
\begin{subfigure}{0.24\textwidth}
 \centering
 \includegraphics[width=2.5cm, height=1.75cm]{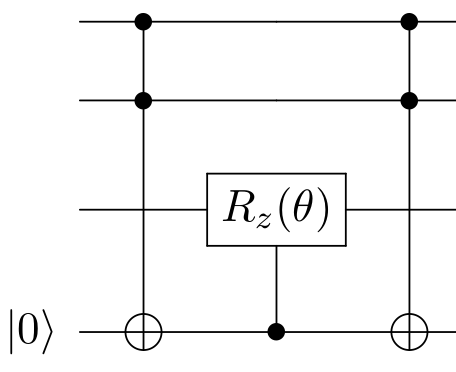}
 \caption{}
 \label{fig:ccRz}
 \end{subfigure}
 \caption{$cR_z(\theta)$ (a) implemented by using two Fredkin gates and one $R_z(\theta)$ gate \cite{2013_KMM}, where upper qubit is the control, middle qubit is the target and bottom qubit is the ancilla; and (b) implemented by two CNOT and two $R_z$ gates. (c) $cR_k$ implemented by two Toffoli gates and one $R_k$ gate. (d) $ccR_z(\theta)$ implemented by two Toffoli, one $cR_z(\theta)$ and one ancilla set to $\ket{0}$.}
 \label{fig:cRz}
\end{figure}
Then an upper bound was given by adding the T-count of component unitaries. For example, in Figure \ref{fig:cRz} we have shown two implementations of $cR_z(\theta)$ gate, that we found in literature. In Figure \ref{fig:cRz1} two Fredkin gates (T-count=7 \cite{2020_MM}), one $R_z(\theta)$ and an extra ancilla \cite{2013_KMM} is used. In Figure \ref{fig:cRz2}, the implementation uses two $R_z$ gates. So upper bound on the T-count of $cR_z(\theta)$, averaged over $\theta$ is as follows.
\begin{eqnarray}
 \#\T (cR_z(\theta)) &=& 3.067\log_2\left(1/\epsilon\right)-4.322+14 = 3.067\log_2\left(1/\epsilon\right)+9.678\quad [\text{Figure }\ref{fig:cRz1}]
 \label{eqn:T} \nonumber \\
 \#\T (cR_z(\theta)) &=& 2\left(3.067\log_2\left(1/\epsilon\right)-4.322\right) = 6.134\log_2\left(1/\epsilon\right)-8.644    \quad [\text{Figure }\ref{fig:cRz2}]  \nonumber
 \label{eqn:T2}
\end{eqnarray}
The first upper bound is better (gives less T-count) for every $\epsilon < 0.016$, but the implementation uses an extra ancilla. In Figure \ref{fig:cRk} and \ref{fig:ccRz} we show an implementation of $cR_k$ and $ccR_z(\theta)$ respectively. The latter circuit can be used to implement $ccR_k$ by replacing the $cR_z(\theta)$ with $cR_k$. Upper bounds on T-count can be deduced in a similar manner for the respective unitaries.

We took $\theta=\frac{2\pi}{2^k}$ and varied $k$ from 2 to 11. We were more interested in synthesizing multi-qubit unitaries, since these were not T-count-optimally synthesized before. It took on an average ~48 mins to synthesize a 2-qubit unitary with T-count at most 7; and about 5.7 hours for a 3-qubit unitary with T-count at most 4. We have synthesized only one 2-qubit unitary with T-count 9. We have not synthesized 2 and 3 qubit unitaries with higher T-count because of time constraint. We made the following observations.
\begin{enumerate}
 \item The T-count of controlled rotation gates reduce, as we increase the number of controls, at least for many of the angles and precision tested by us. This has been shown in Table \ref{tab:ccR}. The average running time has been stated before.
 
  \item The T-count of 2-qubit QFT is equal to the T-count of $R_2$ and has been shown in Table \ref{tab:qft}. In this table we have also shown the T-count of 3-qubit QFT for some precision. The running time for these tests has been explicitly mentioned \footnote{In a previous version and also in the published version of this paper we have erroneously reported that the T-count of 2-qubit QFT at precision $10^{-18}$ is 9. This is due to an error in the implementation. T-count of 2-qubit QFT remains 3 even when $\epsilon=10^{-18}$. The authors thank Marc Davis for pointing this out.}.
 
 \item The T-count of $Givens(\theta)$ is similar to $cR_z(\theta)$, on an average. (So we have not shown it separately.)
 
 \item The T-count of $R_z(\theta)$ (and hence $R_k$) agrees with the results given in \cite{2015_KMM}.
\end{enumerate}
The numerical results of this subsection, together with instructions on how to reproduce them, are available online at \url{https://github.com/vsoftco/approx-t}.
\begin{table}[!htbp]
 \centering
 \begin{tabular}{|c|c|c|c|c|c|c|}
  \hline
 $\epsilon$ & $k$ & $R_z(\frac{2\pi}{2^k})$ & $cR_z(\frac{2\pi}{2^k})$ & $cR_k$ & $ccR_z(\frac{2\pi}{2^k})$  & $ccR_k$ \\
  \hline\hline
  \multirow{5}{*}{$0.05$} & 2 & 0 & 2 & 3 & $> 4$ & $> 4$  \\
  \cline{2-7}
   & 3 & 1 & $> 7$ & $> 7$ & $> 4$ & $> 4$  \\
  \cline{2-7}
   & 4 & 8 & $> 7$ & $> 7$ & 0 & 0  \\
  \cline{2-7}
   & 5 & 9 & 0 & 0 & 0 & 0  \\
  \cline{2-7}
  & 6-11 & 0 & 0 & 0 & 0 & 0 \\
  \hline\hline
   \multirow{7}{*}{$10^{-2}$} & 2 & 0 & 2 & 3 & $> 4$ & $> 4$ \\
  \cline{2-7}
   & 3 & 1 & $> 7$ & $> 7$ & $> 4$ & $> 4$   \\
  \cline{2-7}
   & 4 & 18 & $> 7$ & $> 7$ & $> 4$ & $> 4$   \\
  \cline{2-7}
   & 5 & 17 & $> 7$ & $> 7$ & $> 4$ & $> 4$   \\
  \cline{2-7}
& 6 & 16 & $> 7$ & $> 7$ & 0 & 0  \\
  \cline{2-7}
  & 7 & 11 & 0 & 0 & 0 & 0 \\
  \cline{2-7}
  & 8-11 & 0 & 0 & 0 & 0 & 0 \\
  \hline\hline 
  \multirow{2}{*}{$10^{-3}$} & 10 & 26 & $> 7$ & $> 7$ & 0 & 0 \\
  \cline{2-7}
  & 11 & 26 & 0 & 0 & 0 & 0 \\
  \hline
 \end{tabular}
 \caption{Comparison of $\epsilon$-T-count of different (controlled) rotation gates for various angles and precision. }
 \label{tab:ccR}
\end{table}

\begin{table}[!htbp]
\centering
\begin{tabular}{|c|c|c|c|}
 \hline
 $\#$Qubits & $\epsilon$ & $\epsilon$-T-count & Running time \\
 \hline\hline
 \multirow{1}{*}{2} &  $0.05 - 10^{-17}$ & 3 & 35 ms  \\
 \hline\hline
 \multirow{2}{*}{3} & $10^{-2}$ & 0 & 3 ms  \\
 \cline{2-4}
  & $10^{-3}$ & $> 4$ & 24 hr  \\
 \hline
 \end{tabular}
 \caption{$\epsilon$-T-count of 2 and 3 qubit QFT. }
 \label{tab:qft}
\end{table}

\section{Discussion : T-depth-optimal synthesis}
\label{sec:depth}

The algorithms \ref{alg:min} ($\mathcal{A}_{MIN}$) and \ref{alg:decide} ($\mathcal{A}_{DECIDE}$) can be used for T-depth-optimal-synthesis of any multi-qubit unitary, since we know there is a finite generating set \cite{2021_GMM} such that any exactly implementable unitary can be written as a product of elements from this set and a Clifford. We first give some definitions.

\subsubsection*{T-depth of circuits}

Suppose the unitary $U$ implemented by a circuit is written as a product $U=U_mU_{m-1}\ldots U_1$ such that each $U_i$ can be implemented by a circuit in which all the gates can act in parallel or simultaneously. We say $U_i$ has depth 1 and $m$ is the \emph{depth} of the circuit. The \emph{T-depth of a circuit} is the number of unitaries $U_i$ where the $\T/\T^{\dagger}$ gate is the only non-Clifford gate and all the $\T/\T^{\dagger}$ gates can act in parallel\footnote{The remaining Clifford gates within each $U_i$ may not act in parallel.}. 

\subsubsection*{T-depth of exactly implementable unitaries}

The \emph{T-depth} or \emph{min-T-depth of an exactly 
synthesizable unitary} $U\in\clifft_n$, denoted by $\tcount_d(U)$, is the minimum T-depth of a Clifford+T circuit that implements it (up to a global phase). 
In \cite{2021_GMM} a subset, $\mathbb{V}_n\subset\{\prod_{i\in [n]}C\tset_{(i)}C^{\dagger}, C\in\cliff_n,\tset\in\{\T,\T^{\dagger},\id\}\}$, of T-depth-1 unitaries has been defined. 
It has been shown that $|\genv_n|\leq n\cdot 2^{5.6n}$ and any T-depth-1 unitary $U_1\in\clifft_n$ can be written as
\begin{eqnarray}
 U_1=e^{i\phi}\left(\prod_{i\geq 1} V_i\right)C_0\qquad \text{where } V_i\in\genv_n\text{ and }C_0\in\cliff_n
 \label{eqn:TdepthDecompose}
\end{eqnarray}
We call each $V_i$, with T-depth 1, as a \emph{(parallel) block} and it can be written as product of $R(P)$ or $R^{\dagger}(P)$, where $P\in\pm\pauli_n$. It is possible to multiply consecutive T-depth-1 unitaries to get another T-depth-1 unitary (conditions given in \cite{2021_GMM}). Thus $\genv_n$ can be regarded as a generating set (modulo Clifford) for the set of T-depth-1 unitaries, and hence for the complete group $\clifft_n$.
A decomposition which has the minimum number of T-depth-1 unitaries is called a \emph{T-depth-optimal decomposition}. A circuit implementing $U\in\clifft_n$ with the minimum T-depth is called a \emph{T-depth-optimal circuit}.

 \subsubsection*{$\epsilon$-T-depth of approximately implementable unitaries}
 
 The \emph{$\epsilon$-T-depth} of an approximately implementable unitary $W\in\mathcal{U}_n$, denoted by $\tdeps(W)$, is equal to $\tcount_d(U)$, the T-depth of an exactly implementable unitary $U\in\clifft_n$ such that $d(U,W)\leq\epsilon$ and $\tcount_d(U)\leq\tcount_d(U')$ for any $U'\in\clifft_n$ and $d(U',W)\leq\epsilon$.
We call a T-count-optimal (or T-depth-optimal) circuit for any such $U$ as the \emph{$\epsilon$-T-count-optimal} (or \emph{$\epsilon$-T-depth-optimal} respectively) circuit for $W$. 

\subsubsection*{Modification of $\mathcal{A}_{MIN}$}

Since the set $\genv_n$ is finite, so it is not hard to see that algorithms $\mathcal{A}_{MIN}$ and $\mathcal{A}_{DECIDE}$ can be applied to find T-depth-optimal decomposition of any unitary $W\in\mathcal{U}_n$. Replace step \ref{Adec:S} of Algorithm \ref{alg:decide} by $\mathcal{S}\leftarrow \genv_n$. Suppose $W=UE$ for some exactly implementable unitary $U$ such that $\tdeps(W)=\tdepth(U)$ and $d(W,U)\leq\epsilon$. Then we can decompose $U$ as in Equation \ref{eqn:TdepthDecompose}. If we have guessed the correct $V_1,\ldots,V_d$ then after multiplying $\prod_iV_i$ with $W^{\dagger}$ we are left with $EC_0$. Now the amplitude test and conjugation test can be applied to check if we have the correct guess. We have said before that it is possible to multiply consecutive $V_i$ such that the product has T-depth 1. In that case the $\epsilon$-T-depth is less than $d$. So to find the minimum possible T-depth we might have to iterate more than $\tdeps(W)$ times.
\subsubsection*{Time complexity}

The time complexity of conjugation test $\mathcal{A}_{CONJ}$ is same as before. The analysis of the complexity of $\mathcal{A}_{DECIDE}$ is also similar, but here $|\mathcal{S}|=|\genv_n|$, so if we take all possible $m'$-length product at step 2, then the number of iterations for the outer loop 2-14 is at most $n2^{5.6nm'}$. The complexity of the remaining steps are same, so the overall complexity of $\mathcal{A}_{DECIDE}$ is $O(n2^{5.6nm'+4n})$. We explained before that it is possible to combine more than one consecutive unitaries $V_i\in\genv_n$ such that we get one T-depth 1 unitary. Thus this procedure gives a T-depth $m \leq m'$. We do not know how much is the difference $m'-m$.

Alternatively, we can do a pessimistic analysis of $\mathcal{A}_{DECIDE}$. This algorithm is basically an exhaustive search procedure to test for a certain T-depth $m$. Let in step 2 we make sure that we have a T-depth-m unitary $\tilde{U}$ i.e. it is not possible to combine any further. Basically, this means $\tilde{U}$ is from the set of T-depth-1 unitaries modulo Clifford. Now there can be at most $O(4^{n^2})$ of these. This is because there can be at most $O(4^{n^2})$ $n$-length product of $R(P)$. This is a naive bound and more explanations can be found in \cite{2021_GMM}. So this time the outer loop can occur at most $O(4^{n^2m})$ times. Arguing in the same way, the complexity of $\mathcal{A}_{DECIDE}$ is $O(2^{2n^2m+4n})$, and hence complexity of $\mathcal{A}_{MIN}$ is $O(2^{2n^2\tdeps(W)+4n})$.

\subsubsection*{Space complexity}

In step 1 of $\mathcal{A}_{DECIDE}$ we can store $|\genv_n|$ in a symbolic way, for example, for each $V_i=\prod_j R(P_j)\in\genv_n$, simply store the Paulis in the product. Then we can calculate the necessary matrices whenever necessary by taking products of the corresponding $R(P)s$. In all other steps we store $N\times N$ matrix, taking at most $O(N^2)\in O(2^{2n})$ space. Thus space complexity is $O(n2^{5.6n})$. As explained before this approach leads to more running time, without affecting the asymptotic time complexity.

In this paper we do not implement our algorithm to determine $\epsilon$-T-depth. For small enough $\epsilon$, we can use the procedure chalked out in Section \ref{subsec:cktSynth} to synthesize a T-depth-optimal circuit.

\section{Conclusion}
\label{sec:conclude}

In this paper we have developed results and algorithms that can be used to determine the optimal count or depth of non-Clifford gates required to implement any multi-qubit unitary with the gates of a finite universal gate set, that consists of Clifford and non-Clifford gates like Clifford+V, Clifford+CS, etc. Our primary focus has been the Clifford+T set. Given an $2^n\times 2^n$ unitary $W$, the space complexity of our algorithm is $\poly(2^n)$. The time complexity has an exponential dependence on $\teps(W)$ or $\tdeps(W)$, while determining T-count and T-depth-respectively. For small enough $\epsilon$, we show how we can synthesize the $\epsilon$-T-count or $\epsilon$-T-depth-optimal circuit. 

To the best of our knowledge, we give the first algorithm that works for arbitrary multi-qubit unitaries (exactly or approximately implementable) and any universal gate set. We think it is unlikely that we can have an algorithm whose complexity has a polynomial dependence on $\epsilon$-T-count or $\epsilon$-T-depth. A complexity theoretic argument to prove or disprove Conjecture \ref{conj:approx} will throw theoretical insights into the complexity of these problems. It might be possible to decrease the exponent in the time complexity by applying techniques like meet-in-the-middle \cite{2014_GKMR} or nested meet-in-the-middle \cite{2020_MM, 2021_GMM}. It is not hard to see that our algorithm can be parallelized. It will be interesting to investigate if additional tricks can be used. Another interesting question is to find more compact generating sets for other universal gate sets for multi-qubit unitaries.

\section*{Acknowledgement} 
We thank Vern I. Paulsen and Adina Goldberg for useful discussions. We thank Earl Campbell and Nathan Wiebe for pointing out the (previous) implementations of $cR_z(\theta)$ gate (Figure 1). We thank the anonymous reviewers for many helpful comments that helped us improve the manuscript and also for pointing out a mistake in the pseudocode (Algorithm 3). We also thank Jiaxin Huang and Hong Tao Zhang for pointing out mistake in the pseudocode, as well as running some of the tests on their laptops. We thank Marc Davis for pointing out a mistake in the implementation of 2-qubit QFT at precision $10^{-18}$. The authors wish to thank NTT Research for their financial and technical support. This work was supported in part by Canada's NSERC. Research at IQC is supported in part by the Government of Canada through Innovation, Science and Economic Development Canada. The Perimeter Institute (PI) is supported in part by the Government of Canada and Province of Ontario (PI). 

\section*{Author contributions}
The ideas were given by P.Mukhopadhyay. The software implementations were done by V. Gheorghiu. All the authors contributed to the preparation of the manuscript. 

\section*{Data availability}

Numerical results together with instructions on how to reproduce them, are available online at \url{https://github.com/vsoftco/approx-t}.

\section*{Code availability}

The code is available from the corresponding author on request.

\section*{Competing interests}
There are no competing interests.


\appendix
\section{Some additional preliminaries}
\label{app:prelim}
\subsection{Cliffords and Paulis}
\label{app:clifford}

The \emph{single qubit Pauli matrices} are as follows:
\begin{eqnarray}
 \X=\begin{bmatrix}
     0 & 1 \\
    1 & 0
    \end{bmatrix} \qquad  
 \Y=\begin{bmatrix}
     0 & -i \\
     i & 0
    \end{bmatrix} \qquad 
 \Z=\begin{bmatrix}
     1 & 0 \\
     0 & -1
    \end{bmatrix}\nonumber
\label{eqn:Pauli1}
\end{eqnarray}
Parenthesized subscripts are used to indicate qubits on which an operator acts For example, $\X_{(1)}=\X\otimes\id^{\otimes (n-1)}$ implies that Pauli $\X$ matrix acts on the first qubit and the remaining qubits are unchanged.

The \emph{$n$-qubit Pauli operators} are :
$
 \pauli_n=\{Q_1\otimes Q_2\otimes\ldots\otimes Q_n:Q_i\in\{\id,\X,\Y,\Z\} \}.
$

The \emph{single-qubit Clifford group} $\cliff_1$ is generated by the Hadamard and phase gates :
$
 \cliff_1=\braket{\had,\phase} 
 $
where
\begin{eqnarray}
 \had=\frac{1}{\sqrt{2}}\begin{bmatrix}
       1 & 1 \\
       1 & -1
      \end{bmatrix}\qquad 
 \phase=\begin{bmatrix}
       1 & 0 \\
       0 & i
      \end{bmatrix}\nonumber
\end{eqnarray}
When $n>1$ the \emph{$n$-qubit Clifford group} $\cliff_n$ is generated by these two gates (acting on any of the $n$ qubits) along with the two-qubit $\CNOT=\ket{0}\bra{0}\otimes\id+\ket{1}\bra{1}\otimes\X$ gate (acting on any pair of qubits). 

The Clifford group is special because of its relationship to the set of $n$-qubit Pauli operators. Cliffords map Paulis to Paulis, up to a possible phase of $-1$, i.e. for any $P\in\pauli_n$ and any $C\in\cliff_n$ we have
$
    CPC^{\dagger}=(-1)^bP'
$
for some $b\in\{0,1\}$ and $P'\in\pauli_n$. In fact, given two Paulis (neither equal to the identity), it is always possible to efficiently find a Clifford which maps one to the other.
\begin{fact}[\cite{2014_GKMR}]
 For any $P,P'\in\pauli_n\setminus\{\id\} $ there exists a Clifford $C\in\cliff_n$ such that $CPC^{\dagger}=P'$. A circuit for $C$ over the gate set $\{\had,\phase,\CNOT\}$ can be computed efficiently (as a function of $n$).
 \label{fact:cliffConj}
\end{fact}

\begin{fact}
Let $Q=\sum_{P\in\pauli_n}q_PP$ be the expansion of a matrix $Q$ in the Pauli basis. Then
$$q_P=\tr\left(QP\right)/N \qquad [N=2^n].
$$
Further, if $Q$ is a unitary then
$$
    \sum_{P\in\pauli_n}\left|q_P\right|^2=1
$$
 \label{app:fact:trCoeff}
\end{fact}
\begin{proof}
 If $Q=\sum_{P\in\pauli_n}q_PP$, then taking trace on both sides and dividing by $N$ we get the first equation.
 
 If $Q$ is unitary then we get the following.
 \begin{eqnarray}
  QQ^{\dagger}=\id=\left(\sum_Pq_PP\right)\left(\sum_P\overline{q_P}P\right)=\sum_P|q_P|^2\id+\sum_{P\neq P'}q_P\overline{q_{P'}}PP'    \nonumber
 \end{eqnarray}
Now taking trace on both sides and dividing by $N$, we get the second equation.
\end{proof}

\subsection{The group generated by Clifford and $\T$ gates}
\label{app:Jn}

The group $\clifft_n$ is generated by the $n$-qubit Clifford group along with the $\T$ gate, where
\begin{eqnarray}
 \T=\begin{bmatrix}
     1 & 0 \\
     0 & e^{i\frac{\pi}{4}}
    \end{bmatrix}
\end{eqnarray}

Thus for a single qubit
$
 \clifft_1 = \braket{\had,\T}  
$
and for $n>1$ qubits
\begin{eqnarray}
 \clifft_n=\braket{\had_{(i)},\T_{(i)},\CNOT_{(i,j)}:i,j\in [n]}.
 \nonumber
\end{eqnarray}
It can be easily verified that $\clifft_n$ is a group, since the $\had$ and $\CNOT$ gates are their own inverses and $\T^{-1}=\T^7$. Here we note $\phase=\T^2$.

\section{Some additional upper bound}
\label{app:bound}

In this section we want to prove the following claim.
\begin{claim}
$$
  \sum_{\hat{P}}|\conj{e_{\hat{P}}}||e_{\hat{P'}}|=  \sum_{\hat{P}}|e_{\hat{P}}||e_{\hat{P'}}|\leq 2\epsilon^2-\epsilon^4 \qquad [\hat{P'}\neq\hat{P}, \hat{P},\hat{P'}\neq\id] 
$$
For each $\hat{P}$, there exists one $\hat{P'}$ that satisfies a certain condition. 
 \label{claim:e1e2Bound}
\end{claim}

\begin{proof}

Let us denote $|e_{\hat{P}}|$ by  $x_i$, where $i\in {1,\ldots, N^2-1}$. If $\hat{P'}$ is the Pauli that satisfies the given condition with $\hat{P}$, then we denote the corresponding variable by $x_{i'}$. We must remember that this convention of indices is adapted to identify the variables that satisfy a condition. Each of the variables $x_i \geq 0$. From Equation \ref{eqn:eNotI_1} we know that
$  \sum_{i}x_i^2\leq 2\epsilon^2-\epsilon^4$. 
\begin{eqnarray}
 \sum_{\hat{P}}|e_{\hat{P}}||e_{\hat{P'}}|=\sum_ix_ix_{i'}\leq \sum_{i,j\neq i}x_ix_j \quad [\because x_i\geq 0]    \nonumber
\end{eqnarray}
So we upper bound the following optimization problem.
\begin{eqnarray}
 \max_{x_i} && \sum_{i,j\neq i}x_ix_j \nonumber \\
 \text{such that }&& \sum_{i\neq 0}x_i^2= 2\epsilon^2-\epsilon^4
\end{eqnarray}
We can use Karush-Kuhn-Tucker (KKT) conditions \cite{1939_K, 1951_KT, 2014_KT} to solve the above optimization problem with inequality constraint. For simplicity and without loss of much generality we have used equality in the constraint and follow the Lagrangian method of optimization. 
Let $\lambda$ is a Lagrange multiplier. The Lagrange formulations is:
\begin{eqnarray}
 \mathcal{L}&=&\sum_{i,j\neq i}x_ix_j +\lambda\left(\sum_{i\neq 0}x_i^2 - (2\epsilon^2-\epsilon^4)\right)    \nonumber
\end{eqnarray}
To find the optimal points we have the following for each $i$.
\begin{eqnarray}
 \frac{\partial\mathcal{L}}{\partial x_i}&=&2\sum_{j\neq i}x_j+2\lambda x_i=0   \nonumber \\
 &\implies&\sum_{j}x_j+(\lambda -1)x_i=0    \nonumber\\
 &\implies& x_i=\frac{\sum_{j}x_j}{1-\lambda}
\end{eqnarray}
Since this is true for each $x_i$, we can infer that $x_i=x_j$, where $i\neq j$ and hence
\begin{eqnarray}
 x_i^2=\frac{2\epsilon^2-\epsilon^4}{N^2-1}\quad&\implies& x_i=\sqrt{\frac{2\epsilon^2-\epsilon^4}{N^2-1}}   \nonumber. \\
 \text{and } \sum_{i,j\neq i}x_ix_j &\leq&(N^2-1)\frac{2\epsilon^2-\epsilon^4}{N^2-1}=2\epsilon^2-\epsilon^4.
\end{eqnarray}
The claim follows.

\end{proof}

\section{Circuit for Quantum Fourier Transform}
\label{app:Tcount}

\begin{figure}[h]
\centering
\includegraphics[width=12cm, height=4cm]{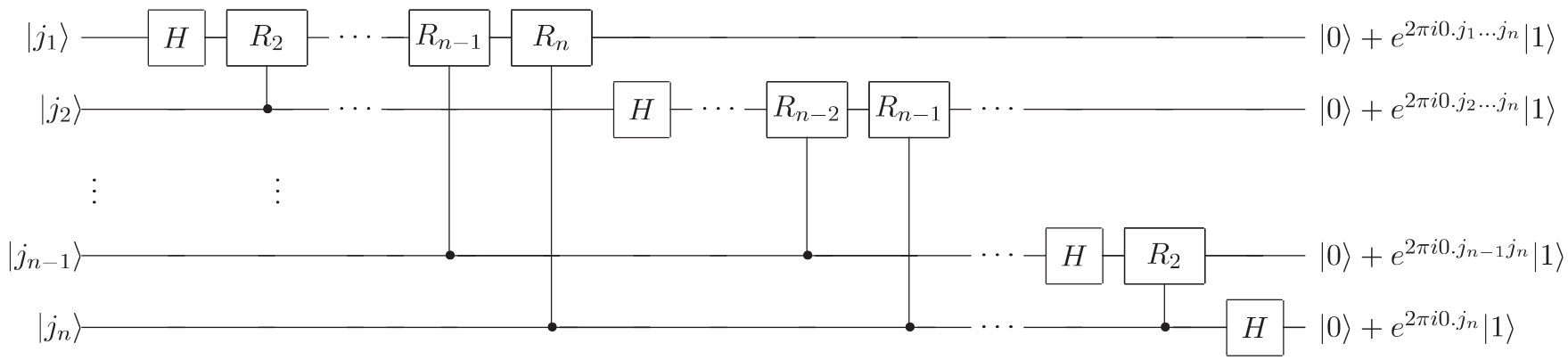}
\caption{Quantum Fourier Transform circuit given in \cite{2010_NC}}
\label{fig:qft}
\end{figure}

\end{document}